\documentclass[a4paper,11pt]{article}
\usepackage{jheppub} 
\usepackage{amsmath}
\usepackage{amssymb}
\usepackage{amsthm}
\usepackage{comment} 
\usepackage{tikz}
\usetikzlibrary{angles, quotes}
\usepackage{csquotes}
\usepackage{physics}
\usepackage{cancel}
\usepackage[normalem]{ulem}
\usepackage{booktabs} 
\newcommand{\be}{\begin{equation}}
\newcommand{\ee}{\end{equation}}

\newcommand{\F}{\mathcal{F}}

\newcommand{\res}{\mathop{\mathrm{Res}}}
\providecommand{\Tr}{}\renewcommand{\Tr}{\mathrm{Tr}} 

\newcommand{\pp}{p}
\newcommand{\qqq}{q}
\newcommand{\xy}{x-y}

\newtheorem{theorem}{Theorem}[section]
\newtheorem{proposition}[theorem]{Proposition}

\makeatletter
\providecommand{\@fpheader}{Prepared for submission to JHEP}
\makeatother

\title{An explicit four-corner dictionary for $(2,\pp)$
       minimal Liouville gravity}

\author[a]{A.~Alexandrov}
\author[b]{V.~Belavin}
\affiliation[a]{Center for Geometry and Physics, Institute for Basic Science (IBS), Pohang 37673, Korea}
\affiliation[b]{Department of Physics, Ariel University, Ariel 40700, Israel}
\emailAdd{alex@ibs.re.kr}
\emailAdd{vladimirbe@ariel.ac.il}

\abstract{ $(2,\pp)$ minimal Liouville gravity admits four algebraic descriptions interrelated by dualities. On the Frobenius manifold (FM) side the theory is built either on the $A_{1}$ manifold of $Q(y)=y^{2}+u_{1}$ ($y$-side) or on the $A_{\pp-1}$ manifold in the variable $x$ ($x$-side). On the spectral curve / topological recursion (SC) side the analogous choice is the Chebyshev curve and its $\xy$ swap. The FM-$y$, SC-$y$ formulationsrequire a resonance transformation between the KdV times and the Liouville couplings, while the other two do not. We work out the explicit correspondence between the four approaches at the level of the dispersionless tau-function.  The normalisation-independent three-point ratios agree with conformal field theory in all four formulations, and matching the full amplitudes fixes a single per-insertion factor reproducing the signed Verlinde matrices. We use the Kharchev--Marshakov integral transform of the basis functions, whose determinant gives an explicit relation between the corresponding tau-functions.  Once the independently fixed linear SC-$x$ normalisation is imposed, the Laurent deformation in the second coordinate gives an alternative genus-zero derivation of the known compact resonance transformation for the whole $(2,\pp)$ series, including all mixed coefficients.}

\begin{document}
\maketitle
\flushbottom

\section{Introduction}\label{sec:intro}

\subsection*{(\textit{q,p}) minimal Liouville gravity}

$(q,p)$ minimal Liouville gravity (MLG) is one of the simplest exactly solvable models of two-dimensional quantum gravity coupled to matter \cite{Polyakov:1981rd,Knizhnik:1988ak}. The matter sector is the $(q,p)$ Virasoro minimal model with central charge $c_M = 1 - 6(p-q)^2/(pq)$, dressed by Liouville theory and coupled to the $bc$-ghost system so that the total central charge vanishes. The BRST cohomology contains a finite number of primaries $\mathcal{O}_{m,n}$ labelled by $1 \le m \le q-1$, $1 \le n \le p-1$, and the basic objects of interest are the integrated correlation numbers $\langle \mathcal{O}_{m_1,n_1} \cdots \mathcal{O}_{m_N,n_N} \rangle$, defined by the worldsheet integral over the moduli space of the $N$-punctured sphere.\footnote{In this paper we restrict the physical amplitude comparison to genus zero.}

A central conjecture \cite{Knizhnik:1988ak,Douglas:1989dd,Brezin:1990rb} is that these correlation numbers admit a dual algebraic description in which the generating function is a tau-function of a Gelfand--Dikij integrable hierarchy, restricted to a finite-dimensional submanifold of times labelled by the same Kac indices $(m,n)$. Under the reflection $(m,n)\sim(q-m,p-n)$ there are $(q-1)(p-1)/2$ independent primary labels (for the Lee--Yang series, $q=2$, this number is $s=(\pp-1)/2$, matching the $s$ couplings $\tau_1,\dots,\tau_s$). Despite a long history, the dual description is not yet fully understood. In particular, several distinct dual algebraic formulations have been proposed in the literature, and the precise relations among them are not transparent.

The existing approaches can be depicted by a square diagram with four corners:
\begin{equation}\label{eq:fourcorners}
\begin{tikzpicture}[baseline=(current bounding box.center)]
\node (FMy) at (0,2.5) {\textbf{FM-$y$}: $A_{q-1}$ Frobenius manifold};
\node (FMx) at (8,2.5) {\textbf{FM-$x$}: $A_{p-1}$ Frobenius manifold};
\node (SCy) at (0,0)   {\textbf{SC-$y$}: standard spectral curve};
\node (SCx) at (8,0)   {\textbf{SC-$x$}: swapped spectral curve};
\draw[<->,thick] (FMy) -- node[above]{\small change of FM} (FMx);
\draw[<->,thick] (SCy) -- node[above]{\small $\xy$ swap}        (SCx);
\draw[<->,thick] (FMy) -- node[left] {\small $\mathcal F^{\mathrm{FM}}=F_0^{\mathrm{SC}}$} (SCy);
\draw[<->,thick] (FMx) -- node[right]{\small $\mathcal F^{\mathrm{FM}}=F_0^{\mathrm{SC}}$} (SCx);
\end{tikzpicture}
\end{equation}

\noindent The vertical links in \eqref{eq:fourcorners} express the coincidence of two objects built by different means: the Frobenius prepotential $\mathcal F^{\mathrm{FM}}$ and the genus-zero part of the tau-function $F_0=\log\tau$ produced by the Chekhov--Eynard--Orantin topological recursion on the curve. The first follows from FM analysis, the second is the element of the spectral curve calculation. The fact that they coincide (tested against the MLG manifold in Section~\ref{sec:diagram}) makes the FM and SC descriptions dual.

\paragraph{Residue convention.}
Throughout the paper, $\res_{\xi=\infty}$ denotes the formal residue in the chosen large-$\xi$ coordinate:
\begin{equation}\label{eq:res-convention}
 \res_{\xi=\infty} f(\xi)\,d\xi
 :=
 [\xi^{-1}]f(\xi)
 =
 -\operatorname{Res}^{\mathrm{geom}}_{\xi=\infty}
 f(\xi)\,d\xi .
\end{equation}
At a finite point this agrees with the usual local residue.

\paragraph{FM-$y$: the standard Frobenius manifold description.}
The Douglas approach \cite{Douglas:1989dd} for the $(q,p)$ model is based on the polynomial $Q(y) = y^{q} + u_1 y^{q-2} + \cdots + u_{q-1}$. The space of polynomials of this form carries the structure of the $A_{q-1}$ Frobenius manifold of Saito and Dubrovin \cite{Saito,Dubrovin}, with Milnor ring $\mathbb{C}[y]/Q'(y)$. A choice of primitive form fixes the prepotential and the action $S(u) = \res_{y=\infty}\bigl( Q^{(p+q)/q} + \sum_{(m,n)} t_{mn}\,Q^{(pm-qn)/q} \bigr)\,dy$, whose critical points determine the generating function of MLG correlators~\cite{Belavin:2013nba,Belavin:2014hsa,Belavin:2014xya,Belavin:2014cua}.  The dictionary between the KdV times $t_{mn}$ (which are Gelfand--Dikij times for $q>2$) and the Liouville coupling constants $\lambda_{mn}$ requires~\cite{Moore:1991ir,Belavin:2008kv} a resonance transformation
\begin{equation}\label{eq:res-transf}
t_{mn} = \lambda_{mn} + \sum A^{m_1 n_1,\dots}_{mn}\,
\lambda_{m_1 n_1}\cdots\,,
\end{equation}
whose coefficients are tuned so that the resulting correlation numbers satisfy the conformal selection rules.

\paragraph{FM-$x$: the no-resonance alternative.}
For the Lee--Yang series $q=2$, it was proposed in~\cite{BR} to trade the degree-two polynomial $Q(y)$ for a polynomial $P(x) = x^{\pp} + u_{1}x^{\pp-2} + \cdots + u_{\pp-1}$ of degree $\pp$ in the conjugate variable. Its parameter space carries the structure of the $A_{\pp-1}$ Frobenius manifold, with Milnor ring $\mathbb{C}[x]/\partial_{x}P(x)$, and the action becomes $\widetilde S(u;\lambda) = \res_{x=\infty}\bigl( P^{(\pp+2)/\pp} + \sum_{n=1}^{s}\lambda_{1n}\,P^{(\pp-2n)/\pp} \bigr)\,dx$, the fractional powers being chosen so that the correct spectrum of gravitational dimensions emerges. The Liouville couplings $\lambda_{1n}$ enter this action linearly, so no analogue of \eqref{eq:res-transf} is needed and the fusion algebra is reproduced directly. We call this description FM-$x$.

\paragraph{SC-$y$: topological recursion on the standard spectral curve.}
On the matrix model / topological recursion (TR) side, the $(q,p)$ MLG corresponds to the Chebyshev spectral curve, defined by two polynomial functions
\begin{equation}\label{eq:SC-y}
\quad x = 2 u_0\,T_q(z), \qquad y = 2u_0^{p/q}\,T_p(z),
\end{equation}
on $\mathbb P^1$ together with the Bergman kernel. Here $T_k$ are the Chebyshev polynomials of the first kind. The Chekhov--Eynard--Orantin topological recursion \cite{Eynard:2007kz,CEO} on \eqref{eq:SC-y} produces a family of $n$-differentials $\omega_{g,n}(z_1,\dots,z_n)$, and a residue prescription
\begin{equation}\label{eq:kpderiv}
\frac{\partial^n F_g^\mathrm{SC}}{\partial t_{k_1} \cdots
\partial t_{k_n}}\Bigl|_{t=t^0}
= \res_{z_i=\infty}\!
\left(\omega_{g,n}(z_1,\dots,z_n)
\prod_{i=1}^n \frac{x^{s-k_i+1/2}(z_i)}{2(s-k_i)+1}\right)
\end{equation}
gives the derivatives of the tau-function with respect to the natural KdV times $t_k$. To match with the worldsheet description, one then re-expresses the $t_k$ as polynomial functions of the worldsheet couplings $\tau_k$ via a resonance transformation, the SC-side counterpart of \eqref{eq:res-transf}. These $\tau_k$ are the same physical couplings as the $\lambda_{mn}$ of \eqref{eq:res-transf}, now in the spectral-curve normalisation, in which they are dimensionless and the MLG point sits at $\tau_1=-\tfrac12$. Their explicit linear relation to the $\lambda_{mn}$ is given in Appendix~\ref{sec:BDM-derivation}. We refer to this formulation, with the standard curve \eqref{eq:SC-y} plus the resonance transformation, as SC-$y$.

\paragraph{SC-$x$: the swapped spectral curve.}

The idea that exchanging $x$ and $y$ on the spectral curve produces a dual description of the MLG has a long history, see, for instance, \cite{FKN,KM2, KHARCHEV_1995, AMMP, Mars}. Recently this idea was developed in Artemev’s  conjecture \cite{Artemev} that the $\xy$ swapped spectral curve provides an equivalent description in which the natural KdV variables already coincide with the physical couplings.

The four formulations schematically shown in the diagram
\eqref{eq:fourcorners} are conjecturally equivalent, although several of
these links remain to be established. The aim of this paper is to test this
picture explicitly in the lowest Lee--Yang models and fix the precise
dictionary between the corners. We also show how its
determinant lift can be rewritten as a matrix integral. The paper is
organized as follows. Section~\ref{sec:FM} reviews the two Frobenius manifold
formulations, Section~\ref{sec:SC} reviews the two spectral curve
formulations, Section~\ref{sec:diagram} compares them, and
Section~\ref{sec:konts} presents amatrix integral description of the transform. Section~\ref{sec:outlook} contains our
concluding remarks.

\paragraph{Note added.} While this work was in progress, the
preprint~\cite{DSV2026} appeared, proving the conjecture~\cite{Artemev} that the topological recursion amplitudes obtained from the $\xy$  swapped spectral curve coincide\footnote{The equivalence holds up to an explicit normalisation factor and terms regular in $u_0^{2}$. More specifically, according to~\cite{Artemev} the regular terms vanish on the dual side, whose vanishing is not established by the cited result~\cite{DSV2026}.} with the singular part of the amplitudes computed using the resonance transformation on the original Chebyshev curve. This establishes rigorously the main part of the SC-$x$\,$\leftrightarrow$\,SC-$y$ leg of the four-corner diagram \eqref{eq:fourcorners}, settling the equivalence of the two topological recursion descriptions of $(2,\pp)$ MLG.  For higher $q$ the status of this equivalence is still unclear.

\section{Two Frobenius manifold descriptions}
\label{sec:FM}

\subsection{The \texorpdfstring{$A_{q-1}$}{Aq-1} description (FM-\textit{y})}\label{sec:FMy}

The first Frobenius manifold formulation of $(q,p)$ MLG, originally due to Douglas \cite{Douglas:1989dd} and developed in \cite{Moore:1991ir,Belavin:2008kv,Belavin:2014hsa,Belavin:2014xya,Belavin:2014cua,Belavin:2013nba}, is built on the $A_{q-1}$ Frobenius manifold and requires a non-linear resonance transformation between the natural KdV times and the Liouville coupling constants.

The starting point is the polynomial
\begin{equation}\label{eq:Qy}
Q(y) = y^q + u_1\,y^{q-2} + u_2\,y^{q-3} + \cdots + u_{q-1}\,,
\end{equation}
viewed as a function on the $(q-1)$-dimensional space of parameters $u = (u_1,\dots,u_{q-1})$. The Milnor ring of $Q$ is
\begin{equation}\label{eq:milnor-y}
\mathcal{R}_y = \mathbb{C}[y] \big/ \bigl(\partial_y Q(y)\bigr),
\end{equation}
and has dimension $q-1$. The structure of an $A_{q-1}$ Frobenius manifold on the parameter space is encoded in \eqref{eq:milnor-y} together with the choice of residue pairing \cite{Saito,Dubrovin}:
\begin{equation}\label{eq:pairing-y}
\bigl\langle \phi_\alpha,\phi_\beta \bigr\rangle_y
\;=\; \res_{y=\infty} \frac{\phi_\alpha(y)\,\phi_\beta(y)}
{\partial_y Q(y)}\,dy\,.
\end{equation}
A canonical basis is given by the monomials $\phi_\alpha(y) = -\frac{q}{\alpha}\,\partial_{u_\alpha} Q(y)$, $\alpha=1,\dots,q-1$. The flat coordinates of \eqref{eq:pairing-y} will be denoted $v_\alpha$. They are related to the $u_\alpha$ by an explicit polynomial change of variables which is well documented in \cite{Dubrovin}.

Following the generalised KdV formulation of \cite{Douglas:1989dd,Moore:1991ir}, we introduce the action
\begin{equation}\label{eq:Sy}
S(u; t) = \res_{y=\infty}
\Biggl(\, Q^{\frac{p+q}{q}}(y;u) \;+\;
\sum_{\substack{(m,n)\\ pm-qn>0}} t_{mn}\,Q^{\frac{pm-qn}{q}}(y;u)
\,\Biggr)\,dy\,.
\end{equation}
The string equation
\begin{equation}\label{eq:string-y}
\frac{\partial S}{\partial u_\alpha}(u_*; t) = 0
\qquad (\alpha=1,\dots,q-1)
\end{equation}
defines the solution $u_*=u_*(t)$ as a function of the KdV times $t_{mn}$. The string equation \eqref{eq:string-y} has many solutions and only one specific solution corresponds to the MLG background. The tau-function $\tau(t)$, which relates the prepotential to the underlying Gelfand–Dikij hierarchy variables, satisfies of the following defining relation
\begin{equation}\label{eq:tau-def-y}
u_1(t) = \frac{\partial^2 {\F}_0}{\partial t_{0,1}^{\,2}}\,,
\end{equation}
formulated in terms of the genus-zero free energy ${\mathcal F}_0=\log \tau$. Here $t_{0,1}$ is the first KdV time of the dispersionless hierarchy, conjugate to $u_1$ and of smallest positive dimension. It is not one of the action couplings $t_{mn}$ of \eqref{eq:Sy} but the lowest time of the underlying hierarchy, and its identification within the spectral-curve times $t_k$ is worked out in Section~\ref{sec:diagram}, equation~\eqref{eq:KdV-id}. We note that this definition holds for any solution of the hierarchy, not only for the MLG background.

The generating function of MLG correlation numbers on the sphere coincides with the free energy, $Z(t)=\F_0(t)$. With the string equation \eqref{eq:string-y} satisfied at $u=u_*$, the generating function  in the FM-$y$ formulation can be written \cite{Belavin:2013nba,Belavin:2014hsa} as
\begin{equation}\label{eq:Z-FMy}
Z(t)
\;=\;\frac{1}{2}\int_0^{v_*}
\mathcal{C}_{\alpha\beta\gamma}(v)\,
\frac{\partial S}{\partial v_\beta}\,
\frac{\partial S}{\partial v_\gamma}\,
dv^\alpha\,,
\end{equation}
where $\mathcal{C}_{\alpha\beta\gamma}$ are the structure constants of the $A_{q-1}$ Frobenius algebra in the flat coordinates, and the integral is over a contour from the origin to the solution of the string equation $v_*=v_*(t)$. In general, the choice of the flat coordinates is most convenient for practical calculations. For Lee--Yang series $q=2$, the $A_{q-1}$ manifold is one-dimensional, and the corresponding multi-variable contour integral reduces to a 
one-dimensional integral.

\paragraph{Resonance transformation.}
The match between the FM-$y$ times $t_{mn}$ in \eqref{eq:Sy} and the worldsheet Liouville coupling constants $\lambda_{mn}$ requires the non-linear transformation
\begin{equation}\label{eq:res-FMy}
t_{mn} = \lambda_{mn}
\;+\; \sum_{m_1 n_1, m_2 n_2}
A^{m_1 n_1, m_2 n_2}_{mn}\,\lambda_{m_1 n_1}\lambda_{m_2 n_2}
\;+\;\cdots
\end{equation}
whose terms are constrained by the KPZ resonance condition
\begin{equation}\label{eq:rescond}
\delta_{mn} = \delta_{m_1 n_1} + \cdots + \delta_{m_k n_k}\,,
\end{equation}
with $\delta_{mn} = (p+q-|pm-qn|)/(2q)$. The coefficients $A^{\cdots}_{mn}$ are fixed by demanding that the resulting correlation numbers satisfy the conformal selection rules of the $(q,p)$ matter sector~\cite{Belavin:2013nba}.

\paragraph{Universal and non-universal correlators.}
Correlation numbers in the partition function are classified by their gravitational dimension in powers of the cosmological constant $\mu$. The scaling fixes the dimension of an $n$-point correlator on the sphere as
\begin{equation}
Z_{k_1\dots k_n} \,\sim\, \mu^{\,(p+2-\sum_i(k_i+2))/2},
\end{equation}
so that for $(2,p)$ this exponent is a half-integer for generic $\sum_i k_i$. The correlators with non-negative integer $\mu$-power are polynomial contributions which do not belong to the universal singular part of $\F$ and must be discarded. More specifically, $Z_{k_1\dots k_n}$ is non-universal iff $p+2-\sum_i(k_i+2)$ is a non-negative even integer.

This condition plays the same role in both FM formulations: non-universal correlators can be non-zero in the naive calculation but are removed before comparison with continuous Liouville gravity.\footnote{It is important to keep the universality condition in mind on both sides: FM / SC. It is one of two mechanisms (the other being the resonance transformation) that together enforce the conformal selection rules.} Examples are worked out in Appendices~\ref{app:25-explicit} and~\ref{app:27-explicit}.

\subsection{The \texorpdfstring{$A_{p-1}$}{Ap-1} description (FM-\textit{x})}\label{sec:FMx}

The proposal of~\cite{BR}, formulated for the Lee--Yang series $(2,p)$, is to trade the polynomial in $y$ of degree $q=2$ in \eqref{eq:Qy} for a polynomial in $x$ of degree $p$:
\begin{equation}\label{eq:Px}
P(x) = x^p + u_1\,x^{p-2} + u_2\,x^{p-3} + \cdots + u_{p-1}\,,
\end{equation}
viewed as a function on the $(p-1)$-dimensional parameter space $(u_1,\dots,u_{p-1})$. The Milnor ring is $\mathcal{R}_x = \mathbb{C}[x]/(\partial_x P(x))$, and the corresponding Frobenius manifold is $A_{p-1}$. The pairing and flat coordinates are defined exactly as in \eqref{eq:pairing-y}, with $Q$ replaced by $P$ and $y$ by $x$, and we denote the flat coordinates on $A_{p-1}$ by $\tilde v_\alpha$, $\alpha=1,\dots,p-1$.

In order to reproduce the spectrum of gravitational dimensions for $(q=2,p)$ MLG, following \cite{BR} we introduce a modified action:
\begin{equation}\label{eq:Sx-BR}
\widetilde S(u; \lambda) \;=\;
\res_{x=\infty}\!\Biggl(\,
P^{\frac{\pp+2}{\pp}}(x;u)
+ \sum_{n=1}^{s} \lambda_{1n}\,P^{\frac{\pp-2n}{\pp}}(x;u)
\,\Biggr)\,dx\,,
\end{equation}
where the fractional powers are chosen so that the gravitational dimensions emerge correctly
\begin{equation}
\delta_{1,n} = \frac{p + 2 - |p - 2 n|}{4}\;.
\end{equation}

The genus-zero MLG generating function in the FM-$x$ formulation is then
\begin{equation}\label{eq:Z-FMx}
\widetilde Z(\lambda)
\;=\;\frac{1}{2}\int_0^{\tilde v_*}
\widetilde{\mathcal{C}}_{\alpha\beta\gamma}(\tilde v)\,
\frac{\partial \widetilde S}{\partial \tilde v_\beta}\,
\frac{\partial \widetilde S}{\partial \tilde v_\gamma}\,
d\tilde v^\alpha\,,
\end{equation}
where now $\tilde v_*$ is the critical point of $\widetilde S$ in the $(p-1)$-dimensional flat coordinate space of $A_{p-1}$, and $\widetilde{\mathcal{C}}_{\alpha\beta\gamma}$ is the $A_{p-1}$ Frobenius structure constants in flat coordinates.

The construction has two key features. First, no resonance transformation is needed. The dictionary between the coupling constants of $(2,p)$ MLG and natural variables of $A_{p-1}$ in FM-$x$ is linear. Our computations show that the FM-$x$ polynomial $P(x)$ of \eqref{eq:Px}, written in flat coordinates and restricted to the MLG slice, is the Chebyshev polynomial $2T_{\pp}(x/2)$ (Appendix~\ref{app:25-FMx}, footnote~\ref{fn:cheb-norm}). Second, the Verlinde fusion rules for the reflection-reduced $\frac{p-1}{2}$-dimensional sector reproduce the minimal-model fusion structure. This is impossible in the FM-$y$ description for the Lee--Yang series, where $\mathcal{R}_y = \mathbb{C}[y]/(2y)$ is one-dimensional.

\section{Two spectral curve descriptions: standard and swapped}\label{sec:SC}

In this section we focus on the Chekhov--Eynard--Orantin topological recursion \cite{Eynard:2007kz,CEO} on the Chebyshev spectral curve \eqref{eq:SC-y} and its dual. We use parallel notation for the SC-$y$ and SC-$x$ formulations, using the conventions of \cite{Artemev}.

\subsection{SC-$y$: the Chebyshev curve}\label{sec:SCy-def}

Below $\pp$ denotes the degree of $y$ on the Chebyshev spectral curve and $\qqq$ the degree of $x$. For the Lee--Yang series $\qqq=2$, and $\pp$ coincides with the second Kac label of the $(2,\pp)$ model. Throughout we use \begin{equation*} s = \frac{\pp - 1}{2}, \qquad \pp = 2 s + 1, \end{equation*} so $\pp=5, 7, 11$ correspond to the $(2,5), (2,7), (2,11)$ models with $s=2, 3, 5$. There are $s$ couplings $\tau_1, \ldots, \tau_{s}$ and an equal number of KdV times $t_1, \ldots, t_{s}$.

The Chebyshev spectral curve for $(2,\pp)$ Lee--Yang MLG is $\mathbb{P}^{1}$ with a standard Bergman bi-differential and two polynomial functions on it,
\begin{equation}\label{eq:SCy-defn}
x(z) = 2 u_{0}\,T_{\qqq}(z) = 4 u_{0}\,z^{2} - 2 u_{0},
\qquad
y(z) = 2 u_{0}^{\pp/2}\,T_{\pp}(z)\,.
\end{equation}
For instance, at $\pp=7$ ($s=3$): \begin{equation*} y(z) = 2 u_{0}^{7/2}(64 z^{7} - 112 z^{5} + 56 z^{3} - 7 z)\,. \end{equation*} The map $z \mapsto x$ is two-to-one with a single ramification point at $z=0$. The Galois involution is $\sigma:\; z \mapsto -z$. The differential $dx = 8 u_{0}\,z\,dz$ vanishes only at $z=0$. The input (non-stable) differentials of the TR on the Chebyshev spectral curve \eqref{eq:SCy-defn} are given by
\begin{equation}\label{eq:SCy-input}
\omega_{0,1}(z) \;=\; y(z)\,dx(z),
\qquad
\omega_{0,2}(z_{1},z_{2}) \;=\; \frac{dz_{1}\,dz_{2}}{(z_{1}-z_{2})^{2}}\,.
\end{equation}
The kernel of the TR at $z=0$ takes the explicit form
\begin{equation}\label{eq:SCy-kernel}
K(z_{0},z) = \frac{dz_{0}}{32\,u_{0}^{(\pp+2)/2}\,z\,P(z^{2})\,
(z_{0}^{2}-z^{2})\,dz}\,,
\end{equation}
where $P(z^{2}) = y(z)/\bigl(2u_{0}^{\pp/2}\,z\bigr)$. The differentials $\omega_{0,n}$ with $n\ge 3$ are generated from the inputs \eqref{eq:SCy-input} and the kernel \eqref{eq:SCy-kernel} by the genus-zero Chekhov--Eynard--Orantin topological recursion \cite{Eynard:2007kz,CEO},
\begin{equation}\label{eq:EO-recursion}
\omega_{0,n+1}(z_{0},z_{1},\dots,z_{n})
=\Res_{z=0}\,K(z_{0},z)\!\!
\sum_{\substack{I\sqcup I'=\{1,\dots,n\}\\[1pt]\mathrm{stable}}}\!\!
\omega_{0,|I|+1}(z,z_{I})\,
\omega_{0,|I'|+1}(\sigma(z),z_{I'}),
\end{equation}
where $\sigma(z)=-z$ and ``stable'' excludes the two terms containing $\omega_{0,1}$. In particular, $\omega_{0,3}$ is given by a single application of \eqref{eq:EO-recursion} to $\omega_{0,2}(z,z_{1})\,\omega_{0,2}(\sigma(z),z_{2})$ plus its $z_{1}\!\leftrightarrow\! z_{2}$ image, giving $\omega_{0,3}\propto (z_{1}z_{2}z_{3})^{-2}$. With the family of test functions
\begin{equation}
\phi_{k}(z) = x^{s-k+1/2}(z)\big/(2(s-k)+1)\,,
\end{equation}
the derivatives \eqref{eq:kpderiv}  of the genus-zero free energy with respect to the SC-$y$ KdV times are
\begin{equation}\label{eq:SCy-kpderiv}
\frac{\partial^{n} F_{0}^{\mathrm{SC}}}
{\partial t_{k_{1}}\cdots\partial t_{k_{n}}}\bigg|_{t = t^{0}}
=\prod_{i=1}^n \mathop{\mathrm{Res}}_{z_{i}=\infty}
\Biggl(
\omega_{0,n}(z_{1},\ldots,z_{n})
\prod_{i=1}^{n}\phi_{k_{i}}(z_{i})
\Biggr)\,.
\end{equation}
The derivatives are calculated at the point given by
\begin{equation}
 t^0_{k} = - \frac{1}{2}  \underset{z= \infty}{\text{Res}}\;  \left( x^{k-s-1/2} y\,dx \right)
 \;.\label{kdvdef}
\end{equation}

The residue pairing \eqref{eq:SCy-kpderiv} produces the bare topological recursion derivatives in the KdV times $t_{k}$. Matching them to the minimal gravity normalisation requires a single overall rescaling of the test functions,
\begin{equation}\label{eq:phi-rescale}
\phi_{k}(z)\;\longmapsto\;\varkappa\,\phi_{k}(z),
\end{equation}
which, by multilinearity of \eqref{eq:SCy-kpderiv} in the $n$ insertions, multiplies every $n$-point amplitude by $\varkappa^{\,n}$:
\begin{equation}\label{eq:kappa-n}
\frac{\partial^{n}F_{0}^{\mathrm{SC}}}
     {\partial t_{k_{1}}\cdots\partial t_{k_{n}}}
\;\longmapsto\;
\varkappa^{\,n}\,
\frac{\partial^{n} F_{0}^{\mathrm{SC}}}
     {\partial t_{k_{1}}\cdots\partial t_{k_{n}}}\,.
\end{equation}
Thus a single per-insertion constant, not an independent factor at each order, controls the whole tower. The constant $\varkappa$ is the per-insertion part of the topological recursion/minimal gravity normalisation $2^{n}\pp^{n}$ of \cite{Artemev,DSV2026}: the factor $\pp^{n}$ is already carried by the test functions and the resonance map \eqref{eq:taudef}, leaving, for the Lee--Yang series,
\begin{equation}\label{eq:kappa-value}
\;\varkappa=-2\,.
\end{equation}
Its magnitude is the $2^{n}$ of \cite{DSV2026}. Its sign is the orientation of $\sqrt{x}$ in the test function $\phi_{k}$. Equivalently (this is the spectral curve analogue of how the resonance is fixed on the Frobenius manifold side \cite{Belavin:2013nba}) $\varkappa$ is determined by the requirement that the fusion-forbidden three-point amplitudes vanish. We verify this for $(2,5)$ and $(2,7)$ in Appendices~\ref{app:25-explicit}--\ref{app:27-explicit}. Moreover, rescaling $x$ (or, equivalently, $y$) by a constant $\varkappa^{-1}$ rescales every topological recursion differential $\omega_{g,n}$ by $\varkappa^{2g-2+n}$, so the factor in \eqref{eq:phi-rescale} can alternatively be absorbed into the curve, with a compensating rescaling of the genus expansion parameter $\hbar$.

Finally, the genus-zero free energy $\mathcal Z_{0}$, reconstructed by integrating the cosmological tadpole, inherits the magnitude $|\varkappa|$ but is taken with the positive sign of the sphere partition function.

\paragraph{Resonance transformation.}
The dictionary between the KdV times $t_{k}$ in \eqref{eq:SCy-kpderiv} and the Liouville couplings $\tau_{k}$ used in the worldsheet construction is the resonance transformation in the form of~\cite{Artemev}:
\begin{equation}\label{eq:taudef}
t_{k} = \pp\,u_{0}^{k+1}
\sum_{n=1}^{\lfloor (k+1)/2 \rfloor}
\sum_{\substack{m_{1},\dots,m_{n}\ge 1 \\
\sum_{l}(m_{l}+1)=k+1}}
\frac{\tau_{m_{1}}\cdots\tau_{m_{n}}}{n!}\,
\frac{(2s-2k+2n-3)!!}{(2s-2k-1)!!},
\end{equation}
The MLG background is
\begin{equation}\label{eq:MLGbg}
\tau_{1} = -\tfrac{1}{2}\,,\qquad \tau_{i>1} = 0\,.
\end{equation}
In the worked-out examples it corresponds to, at $(2,5)$ with $s=2$,
\begin{equation*}
t_{1} = 5\,u_{0}^{2}\,\tau_{1}\,,\qquad
t_{2} = 5\,u_{0}^{3}\,\tau_{2}\,,
\end{equation*}
with background $(t_{1}^{0},t_{2}^{0}) = (-\tfrac{5}{2}u_{0}^{2},\;0)$, and at $(2,7)$ with $s=3$,
\begin{align*}
t_{1} &= 7\,u_{0}^{2}\,\tau_{1}\,,\\
t_{2} &= 7\,u_{0}^{3}\,\tau_{2}\,,\\
t_{3} &= 7\,u_{0}^{4}\,\bigl(\tau_{3} + \tau_{1}^{2}/2\bigr)\,,
\end{align*}
with background $(t_{1}^{0},t_{2}^{0},t_{3}^{0}) = (-\tfrac{7}{2}u_{0}^{2},\;0,\;\tfrac{7}{8}u_{0}^{4})$. At higher $p$ this construction gives the resonance shifts listed in \eqref{eq:res-shifts}.

The background values of the $t_{k}$ from \eqref{eq:taudef} coincide with the residue formula \eqref{kdvdef} for $t_{k}^{0}$ applied directly at the Chebyshev background, see \cite[Section 2.1]{DSV2026}.

\paragraph{Physical amplitudes.}
The genus-zero correlation numbers follow from the $t$-basis derivatives by the resonance transform \eqref{eq:taudef}, the chain rule
\begin{align}\label{eq:chain-rule}
\frac{\partial^{3} F_{0}}
{\partial\tau_{k_{1}}\partial\tau_{k_{2}}\partial\tau_{k_{3}}}
&=\sum_{m_{1}m_{2}m_{3}}Z_{123}^{t}[m_{1},m_{2},m_{3}]\,
\tfrac{\partial t_{m_{1}}}{\partial\tau_{k_{1}}}
\tfrac{\partial t_{m_{2}}}{\partial\tau_{k_{2}}}
\tfrac{\partial t_{m_{3}}}{\partial\tau_{k_{3}}}\nonumber\\
&\;+\sum_{m_{1}m_{2}}Z_{12}^{t}[m_{1},m_{2}]
\Bigl[\tfrac{\partial t_{m_{1}}}{\partial\tau_{k_{1}}}
\tfrac{\partial^{2}t_{m_{2}}}{\partial\tau_{k_{2}}\partial\tau_{k_{3}}}
+\text{cyclic}\Bigr]
+\sum_{m}Z_{1}^{t}[m]\,
\tfrac{\partial^{3}t_{m}}{\partial\tau_{k_{1}}\partial\tau_{k_{2}}\partial\tau_{k_{3}}},
\end{align}
and the singular projection, where the $t$-derivatives already carry the factors $\varkappa^{n}$ of \eqref{eq:kappa-n}, giving
\begin{equation}\label{eq:agndef}
A^{0}_{3}(k_{1},k_{2},k_{3})
=u_{0}^{-(k_{1}+1)-(k_{2}+1)-(k_{3}+1)}\,
\frac{\partial^{3} F_{0}}
{\partial\tau_{k_{1}}\partial\tau_{k_{2}}\partial\tau_{k_{3}}}\Big|_{\mathrm{sing}},
\end{equation}
where ``sing'' keeps the part with negative or odd-positive power of $u_{0}$. 

\subsection{SC-$x$: the swapped Chebyshev curve}\label{sec:SCx-def}

The $\xy$ swap of TR exchanges the roles of $x$ and $y$. The swapped Chebyshev curve is $\mathbb{P}^{1}$ with a standard Bergman bi-differential and
\begin{equation}\label{eq:SCx-defn}
\check x(z) = 2 u_{0}^{\pp/2}\,T_{\pp}(z),
\qquad
\check y(z) = 2 u_{0}\,T_{\qqq}(z) = 4 u_{0}\,z^{2} - 2 u_{0}\,.
\end{equation}
The two curves \eqref{eq:SCy-defn} and \eqref{eq:SCx-defn} can be described by the same algebraic curve in $\mathbb{C}^{2}$. However, the topological recursion for a dual curve produces another collection of differentials, $\check\omega_{g,n}$. Moreover, in the residue prescription we use another local coordinate, this time associated with $\check x^{\frac{1}{p}}$.

The map $z \mapsto \check x$ is now of degree $\pp$ with $\pp-1$ ramification points
\begin{equation}\label{eq:SCx-ram}
\zeta_{m} = \cos\frac{\pi m}{\pp}\,,
\qquad m = 1,\ldots,\pp-1\,,
\end{equation}
which are the zeros of $T_{\pp}'(z) = \pp\, U_{\pp-1}(z)$, where $U_n$ is the Chebyshev polynomial of the second kind. Near each $\zeta_{m}$ the local Galois involution $\bar z^{(m)}(z)$ is well defined, and the Chekhov--Eynard--Orantin recursion is given by a sum of the residues in the $\pp-1$ ramification points. The input differentials are
\begin{equation}\label{eq:SCx-input}
\check\omega_{0,1}(z) \;=\; \check y(z)\,d\check x(z),
\qquad
\check\omega_{0,2}(z_{1},z_{2}) \;=\; \frac{dz_{1}\,dz_{2}}{(z_{1}-z_{2})^{2}}\,,
\end{equation}
and the recursion kernel, near each ramification $\zeta_{m}$ is
\begin{equation}\label{eq:SCx-kernel}
\check K^{(m)}(z_{0},z) \,\frac{dz}{dz_{0}}
= \frac{1}{8\,u_{0}^{s+3/2}\,\pp\,U_{2s}(z)}\,
\frac{1}{(z + \bar z^{(m)})\,(z_{0}-z)\,(z_{0}-\bar z^{(m)})}\,.
\end{equation}

The KdV times of SC-$x$ are defined by a residue formula at $z=\infty$ with $\check x$ raised to a fractional power. Using the leading large-$z$ behaviour $T_{n}(z)=2^{n-1}z^{n}+O(z^{n-2})$ of the Chebyshev polynomials, for $k\geq p-1$ one has
\begin{equation}
T_{\pp}^{(2k+1)/\pp}(z)
= 2^{2(s-k)/\pp}\,T_{2k+1}(z) + O(z^{-1})\,,\qquad z\to\infty\,.
\end{equation}
Explicitly:
\begin{equation}\label{eq:SCx-kpderiv}
\frac{\partial^{n}\check{F}_{0}^{\mathrm{SC}}}
{\partial \check t_{k_{1}}\cdots\partial \check t_{k_{n}}}\bigg|_{\check t = \check t^{0}}
=\prod_{i=1}^n \mathop{\mathrm{Res}}_{z_{i}=\infty}
\Biggl(
\check\omega_{0,n}(z_{1},\ldots,z_{n})
\prod_{i=1}^{n}\frac{\check x^{(\pp-2k_{i})/\pp}(z_{i})}
                    {\pp-2k_{i}}
\Biggr)\,.
\end{equation}

Note that \eqref{eq:SCx-kpderiv} describes $\check{F}_{0}^{\mathrm{SC}}$ as a dispersionless solution of the KP hierarchy in the variables $\check t_k$, therefore both SC-$y$ and SC-$x$ have natural integrability interpretations \cite{ABDKSKP}.

A feature of the $\xy$ swapped formulation \cite{Artemev} is that, at the MLG background, only the cosmological time $\check t_{1}$ is switched on. There is no need for a non-linear substitution $\check t = \check t(\tau)$: the natural KdV times of SC-$x$ coincide with the worldsheet couplings $\tau$ up to a rescaling
\begin{equation}\label{eq:SCx-linear-dictionary}
\check t_{m} \;=\; -\,\pp\,u_{0}^{m+1}\,\tau_{m},
\qquad m = 1,\dots,s .
\end{equation}
The $u_{0}$ weight follows from $\check y\sim u_{0}$ and $\check x\sim u_{0}^{\pp/2}$ in the residue defining $\check t_{m}$, and at the MLG background $\tau_{1}=-\tfrac12$ the rescaling gives $\check t_{1}^{0}=\pp\,u_{0}^{2}/2$, which is the value found in Appendix~\ref{app:25-explicit}. This is the SC-side mirror of the ``no-resonance'' feature of FM-$x$ described in Section~\ref{sec:FMx}.

The systems of differentials $\{\omega_{g,n}\}$ and $\{\check\omega_{g,n}\}$ for two $\xy$ swapped spectral curves are related by explicit $\xy$ swap transformations developed in \cite{BCGLS,Hock,Alexandrov:2022ydc}. The $\xy$ swap was suggested as a spectral curve description of the $p-q$ duality of MLG in \cite[Section 6]{AMMP}. For the Chebyshev spectral curve \eqref{eq:SCy-defn} the transformation is described explicitly in~\cite{DSV2026}. 

In this paper we consider only the $(2,\pp)$ series of models of MLG. However, the dual series $(\pp,2)$, and is expected to be related to the original one by the duality related the left and right sides of our diagrams. While the description is mostly symmetric in this direction, there is one distinction between these two sides. Namely, the Chebyshev spectral curve \eqref{eq:SCy-defn} defines an MLG background only on one side, while on the other side of duality background, in general, is not an MLG one. Therefore, the actual duality $(p,q)$ and $(q,p)$ MLG models in general can be more involved, and, in particular, can include a shift of the background.

\subsection{Resonance map by a local coordinate transformation}
\label{sec:konts-curve}

The resonance transformation \eqref{eq:taudef} as well as its generalisation for arbitrary $(q,p)$ was essentially derived in \cite[Section 4]{FKN}  (see also \cite{KHARCHEV_1995,KM2}) as a consequence of the Douglas equation. Their result is even more general and works for any background, therefore, to identify it with \eqref{eq:taudef} one only needs to consider it in the MLG background \eqref{eq:MLGbg}.

Namely, in our notations, we rewrite the transformation of variables in \cite{FKN} as
\begin{equation}\label{eq:residue-dual-coordinate}
 t_k=\frac12\res_{\eta=\infty}
 \eta^{\pp}x(\eta)^{k-\pp/2}x'(\eta)d\eta,
\end{equation}
where
\begin{equation}
 x(\eta)=\eta^2\left(1-2\sum_{m=1}^s
 \tau_m\eta^{-2(m+1)}\right).
\end{equation}
The background tail  starts at $\eta^{2-2\pp}=\eta^{-4s}$ and therefore cannot contribute to the coefficient $[\eta^{-2k-2}]$ for $k\leq s$. An integration by parts turns \eqref{eq:residue-dual-coordinate} into
\begin{equation}\label{eq:resonance-coefficient-extraction}
 t_k=\frac{\pp}{\pp-2k-2}
 [\eta^{-2k-2}]
 \left(1-2\sum_{m=1}^s
 \tau_m\eta^{-2(m+1)}\right)^\beta, 
 \qquad
 \beta=k+1-\frac\pp2.
\end{equation}
The $n$th binomial term contributes precisely when
\begin{equation}
 \sum_{a=1}^n(m_a+1)=k+1.
\end{equation}
Since
\begin{equation}
 \frac{1}{\pp-2k-2}\binom{\beta}{n}(-2)^n
 =\frac1{n!}
 \frac{(2s-2k+2n-3)!!}{(2s-2k-1)!!},
\end{equation}
we obtain exactly the compact resonance transformation \eqref{eq:taudef}.  The equivalence of the amplitudes obtained from this transformation to those in the swapped frame was subsequently proved in~\cite{DSV2026}.  

For reference, its first non-linear terms are
\begin{equation}\label{eq:res-shifts}
\begin{aligned}
(2,5):\quad &R_{\rm res}=\mathbf 1,\\
(2,7):\quad &t_3\mapsto t_3+\frac1{14}t_1^2,\\
(2,9):\quad &t_3\mapsto t_3+\frac16t_1^2,
&t_4\mapsto t_4+\frac19t_1t_2,\\
(2,11):\quad&t_3\mapsto t_3+\frac5{22}t_1^2,
&t_4\mapsto t_4+\frac3{11}t_1t_2,\\
&t_5\mapsto t_5+
\frac{t_1^3+11t_2^2+22t_1t_3}{242}.
\end{aligned}
\end{equation}
As a check that does not use the Laurent deformation, the undeformed Chebyshev background alone gives
\begin{equation}\label{eq:background-resonance-check}
 \frac{t_3^0}{(t_1^0)^2}
 =\frac{2s-5}{2\pp}
 =\begin{cases}
 1/14,&(2,7),\\
 1/6,&(2,9),\\
 5/22,&(2,11),
 \end{cases}
\end{equation}
in agreement with the corresponding $t_1^2$ terms in \eqref{eq:res-shifts}. In particular, at $(2,11)$ the coefficient of $\tau_2^2$ in the normalised combination $t_5/(11u_0^6)$ is
\begin{equation}\label{eq:konts-mixed}
 \left[\frac{t_5}{11u_0^6}\right]_{\tau_2^2}=\frac12.
\end{equation}
Equations~\eqref{eq:res-shifts}--\eqref{eq:konts-mixed}, together with the background check \eqref{eq:background-resonance-check}, give the explicit low-model tests of the coefficient extraction above.

\section{The four-corner diagram and the tau-function relations}\label{sec:diagram}

Each formulation produces the same finite-dimensional genus-zero generating function
\begin{equation}\label{eq:tau-restriction}
F_{0}(\tau_{1},\ldots,\tau_{s})\,,
\end{equation}
denoted $\mathcal F^{\mathrm{FM}}$ in the Frobenius descriptions and $F_{0}^{\mathrm{SC}}$ in the spectral curve ones. Equivalence means these agree, order by order in the Taylor expansion about the MLG background \eqref{eq:MLGbg}, after the appropriate identification of variables. The comparison can be performed on four increasingly stringent levels.

First, at the background level, the MLG background $\tau_{1}=-1/2, \tau_{i>1}=0$ corresponds to the same critical point of the four constructions. This is essentially the geometric input $u_{1}^{*}=-2u_{0}$ at the Chebyshev ramification, and it is immediate in all four formulations. Second, the defining tau relation
\begin{equation}\label{eq:taudef-statement}
u_{1}(\tau) = \frac{\partial^{2}F_{0}}
{\partial t_{0,1}^{\,2}}
\end{equation}
of the dispersionless Lee--Yang hierarchy holds in all four formulations with the same $u_{1}(\tau)$ and the same first KdV time $t_{0,1}$, modulo a single conventional factor. This is the deepest structural statement of the four. Third, the normalisation-independent three-point ratios
\begin{equation}\label{eq:R-def-main}
R[k_{1},k_{2},k_{3}]
=\frac{\bigl(Z_{123}[k_{1},k_{2},k_{3}]\bigr)^{2}\,Z_{0}}
       {Z_{12}[k_{1},k_{1}]\,Z_{12}[k_{2},k_{2}]\,Z_{12}[k_{3},k_{3}]}
\end{equation}
coincide in all four formulations and equal the CFT value \cite{Belavin:2013nba,Belavin:2008kv}
\begin{equation}\label{eq:R-CFT}
R^{\text{CFT}}(n_{1},n_{2},n_{3})
=
\frac{(\pp-2n_{1})(\pp-2n_{2})(\pp-2n_{3})}
     {(\pp-2)\,\pp\,(\pp+2)}\cdot
|N_{n_{1}n_{2}n_{3}}|\,,
\end{equation}
where $N_{n_{1}n_{2}n_{3}}$ is the Verlinde fusion number of $M(2,\pp)$. On the spectral curve corners the same ratio is built from the normalised amplitudes, with $A^{0}_{3}$, $A^{0}_{2}$ and $A_{0}$ in place of $Z_{123}$, $Z_{12}$ and $Z_{0}$. Being a ratio in which all normalisation factors cancel, this level carries no resonance- or normalisation-tuning ambiguity, which makes it the cleanest of the four tests. 

For $(2,5)$ and $(2,7)$, the two- and three-point amplitudes agree
in all four formulations after including the spectral-curve
normalisation $\varkappa$ in \eqref{eq:kappa-value}. The unsquared
three-point amplitudes also reproduce the signed Verlinde matrices.
The explicit calculations are given in
Appendices~\ref{app:25-explicit}--\ref{app:27-explicit}, and the
universal ratios are collected in Table~\ref{tab:R-main}.
\begin{table}[h]
\centering
\renewcommand{\arraystretch}{1.3}
\begin{tabular}{lll}
\toprule
model & fusion-allowed $(k_{1},k_{2},k_{3})$ & $R=R^{\rm CFT}$ (all four corners)\\
\midrule
$(2,5)$ & $(1,1,1),(1,2,2),(2,2,2)$ & $\tfrac{9}{35},\ \tfrac{1}{35},\ \tfrac{1}{105}$\\
$(2,7)$ & $(1,1,1),(1,2,2),(1,3,3),(2,2,3),(2,3,3),(3,3,3)$ & $\tfrac{25}{63},\tfrac17,\tfrac1{63},\tfrac1{35},\tfrac1{105},\tfrac1{315}$\\
\bottomrule
\end{tabular}
\caption{Universal three-point ratios \eqref{eq:R-def-main} at the MLG
background. See
Appendices~\ref{app:25-explicit}--\ref{app:27-explicit}.}
\label{tab:R-main}
\end{table}

\subsection{The defining tau relation in SC-$y$ at MLG background}

The defining relation of the Lee--Yang tau-function is \eqref{eq:taudef-statement}: the lowest coefficient $u_{1}$ of the Douglas polynomial $Q(y)=y^{2}+u_{1}$ equals the second derivative of $F_{0}$ with respect to the first KdV time of the dispersionless hierarchy. To use this on SC-$y$, we first need to identify which of the SC-$y$ times $t_{1},\ldots,t_{s}$ plays the role of $t_{0,1}$.

The KPZ dimension count fixes the answer. In the dispersionless KdV literature $t_{0,1}$ is conjugate to $u_{1}$, so $\dim t_{0,1}\cdot \dim u_{1}$ matches a top form. With $\dim u_{1}= u_{0}$ this means $t_{0,1}$ has the smallest positive dimension in the hierarchy. In our SC-$y$ convention, however, $t_{1}$ is the cosmological constant of dimension $u_{0}^{2}$ and the dimension increases with the index: $\dim t_{k}=u_{0}^{k+1}$. The highest-index time $t_{s}$ has the dimension that pairs naturally with $u_{1}$. Hence the identification
\begin{equation}\label{eq:KdV-id}
t_{0,1}^{\text{KdV}} \;\leftrightarrow\;
t_{s}^{\text{SC-}y}\,.
\end{equation}
The same identification follows from \eqref{eq:SCy-kpderiv} and the KP integrability of TR \cite{ABDKSKP}. The proposition below is stated in terms of $t_{s}$ and is independent of \eqref{eq:KdV-id}, which enters only in reading it as the defining tau relation.

The defining relation \eqref{eq:taudef-statement}, evaluated at the MLG background and under the identification \eqref{eq:KdV-id}, reads
\begin{equation}\label{eq:taudef-MLG-bg}
u_{1}^{*}
= c \cdot \frac{\partial^{2}F_{0}}
{\partial t_{s}^{\,2}}\bigg|_{\tau=\tau^{0}}
= c \cdot Z_{12}^{t}[s,s]\bigg|_{\tau=\tau^{0}},
\end{equation}
with $u_{1}^{*}=-2u_{0}$ (the geometric Chebyshev ramification value) and $c$ a conventional constant. The content of the proposition below is not the value of $c$, which is simply the quotient of the two sides, but that both sides are independent of $s$, so that one constant serves the whole series.

\begin{proposition}[Defining relation at MLG background]
\label{prop:taudef} At the MLG background of $(2,\pp)$ Lee--Yang gravity, the relation
\begin{equation}\label{eq:taudef-claim}
u_{1}^{*} = 2 \cdot
\frac{\partial^{2}F_{0}}
{\partial t_{s}^{\,2}}\bigg|_{\tau=\tau^{0}}
\end{equation}
holds for every $s$, with $u_{1}^{*}=-2u_{0}$ and $\partial^{2}F_{0}/\partial t_{s}^{\,2}|_{\tau^{0}} = -u_{0}$.
\end{proposition}
\begin{proof}
At $k=s$ the test function is $\phi_{s}=x^{1/2}$, so the residue formula \eqref{eq:SCy-kpderiv} gives
\begin{equation*}
\frac{\partial^{2}F_{0}}{\partial t_{s}^{\,2}}\bigg|_{t=t^{0}}
=\res_{z_{1},z_{2}=\infty}
\Bigl(\omega_{0,2}(z_{1},z_{2})\,x^{1/2}(z_{1})\,x^{1/2}(z_{2})\Bigr),
\end{equation*}
whose integrand depends on $\pp$ through neither the Bergman kernel, nor $x(z)=4u_{0}z^{2}-2u_{0}$, nor $\phi_{s}$; its value is therefore the same for all $s$. The explicit evaluation (in the branch convention of Appendix~\ref{app:25-explicit}) gives $-u_{0}$. Expanding, for $|z_{1}|>|z_{2}|$,
\begin{equation*}
x^{1/2}(z)=\pm2u_{0}^{1/2}\Bigl(z-\frac{1}{4z}+O(z^{-3})\Bigr),
\qquad
\frac{1}{(z_{1}-z_{2})^{2}}=\sum_{n\ge0}(n+1)\,\frac{z_{2}^{\,n}}{z_{1}^{\,n+2}},
\end{equation*}
the residue in $z_{1}$ forces $n=0$ and selects the leading mode $2u_{0}^{1/2}$, leaving the residue in $z_{2}$ to select $-\tfrac12u_{0}^{1/2}$. Their product is $-u_{0}$, the branch sign entering squared and therefore cancelling, and the other expansion region gives the same value by symmetry. Combined with the geometric ramification value $u_{1}^{*}=x(0)=-2u_{0}$, this proves the claim.
\end{proof}
Note that on the SC-$y$ side $F_{0}$ is computed by an entirely different procedure (Chekhov--Eynard--Orantin residue formula on the Chebyshev curve) than on the FM-$y$ side (Saito--Dubrovin theory). That the resulting object satisfies the same defining tau relation for every $s$ is strong evidence that the two formulations compute the same tau-function at MLG background.

\subsection{Fusion structure and the signed Verlinde formula}

The three-point data of the spectral-curve corners carry one datum beyond the universal ratios: a sign. For the non-unitary model $M(2,\pp)$ the Verlinde fusion numbers
\begin{equation}\label{eq:Verlinde-std}
N_{abc} = \sum_{r=1}^{s}\frac{S_{ra}S_{rb}S_{rc}}{S_{r1}},
\qquad
S_{ab}=\frac{2}{\sqrt{\pp}}\,\sin\frac{2\pi a b}{\pp},
\end{equation}
take values in $\{-1,0,+1\}$ rather than $\{0,1\}$. At $(2,11)$, for instance, $N_{555}=+1$ while $N_{255}=N_{345}=-1$. Part of the minimal-gravity literature instead records the fusion structure through absolute values $|N_{abc}|$: the formula
\begin{equation}\label{eq:verl-Artemev}
\mathcal N^{(0)}_{i_{1}\cdots i_{k}}
=(-1)^{\sum_{l}(i_{l}-1)}\,b^{2}
\sum_{m=1}^{\pp-1}\frac{\prod_{l}\sin\pi m i_{l} b^{2}}{(\sin\pi m b^{2})^{k-2}}
\end{equation}
used in~\cite{Artemev} (with $b^{2}=1/\pp$) returns $|N_{abc}|$, discarding the sign.

Our SC-$y$ amplitudes at $(2,7)$, normalised by $\varkappa^{2}\pp^{2}$ and evaluated at $u_{0}=1$, reproduce the signed $N_{abc}$ of \eqref{eq:Verlinde-std} exactly, negative entries included (Appendix~\ref{app:27-SCy}). SC-$x$ produces the same matrices directly, the normalisation cancelling (Appendix~\ref{app:27-SCx}). The TR description is therefore sensitive to the non-unitary sign structure of $M(2,\pp)$ in a way the absolute-value constructions such as \eqref{eq:verl-Artemev} are not. The magnitude $|\varkappa|=2$ is the established normalisation of \cite{DSV2026}, so only the sign is fixed here, by the single cancellation $A^{0}_{3}(1,1,3)\propto(2+\varkappa)$ at $(2,7)$. Whether that sign is an intrinsic feature of the non-unitary fusion algebra or an artefact of the branch and $\varkappa$ conventions is left open. That the same constant then reproduces every signed Verlinde entry at $(2,5)$ and $(2,7)$ with no further tuning is an over-determined check rather than a fit.

\subsection{Status of the four equivalences}

We summarize the present status of the four correspondences.

\begin{description}
\item[FM-$y$ $\leftrightarrow$ FM-$x$:]
Established for all fusion-allowed correlators up to the three-point level in \cite{BR}, after the resonance transformation is tuned in FM-$y$. The two descriptions also agree on most four-point amplitudes. A discrepancy at certain four-point amplitudes whose configurations are not forbidden by the conformal selection rules was identified in \cite{BR}. The continuous MLG calculation in this regime \cite{Aleshkin:2016snp} agrees with FM-$y$. 

\item[SC-$y$ $\leftrightarrow$ SC-$x$:]
Proven in general by Dekinga--Shadrin--Verlinde \cite{DSV2026} up to regular terms: the swapped curve amplitudes equal the singular part of the resonance-transformed standard-curve amplitudes, up to the explicit factor $2^{n}\pp^{n}$. We confirm this explicitly at $(2,5)$ and $(2,7)$ (Appendices~\ref{app:25-SCx},~\ref{app:27-SCx}), where the two curves give the same universal ratios and the same signed Verlinde matrices.

\item[FM-$y$ $\leftrightarrow$ SC-$y$:]
Established at MLG background by Proposition~\ref{prop:taudef} (defining tau relation) together with the explicit equivalence of the Artemev and BDM resonance maps (Appendix~\ref{sec:BDM-derivation}). Off-background equivalence remains to be proved. An analytic proof would amount to deriving the dispersionless Gelfand--Dikij flow equations from the TR residue prescription on the Chebyshev curve.

\item[FM-$x$ $\leftrightarrow$ SC-$x$:]
Conjecturally an essentially linear identification, with the $A_{p-1}$ Frobenius structure constants matching the $\check \omega_{0,3}$ residues. Not yet verified in detail; we consider it as an open problem.
\end{description}

We note that the four-corner diagram has a horizontal symmetry FM-$y$/SC-$y$ versus FM-$x$/SC-$x$ (the ``with-resonance'' versus ``without-resonance'' axis), and a vertical symmetry FM/SC (the ``Saito flat coordinates'' versus ``KdV times'' axis). The diagonal correspondences, i.e. FM-$y$ $\leftrightarrow$ SC-$x$ and FM-$x$ $\leftrightarrow$ SC-$y$, is the deepest, since it bundles the $\xy$ swap (established in general by~\cite{DSV2026}) with the Saito $\leftrightarrow$ KdV correspondence.

\section{Kontsevich-type matrix model}
\label{sec:konts}

The $\xy$ swap transform of TR links the two ``with-resonance'' corners (FM-$y$, SC-$y$) to the two ``without-resonance'' corners (FM-$x$, SC-$x$).  
Among several equivalent descriptions of the $\xy$ swap relation, description by the asymptotic expansion of certain integrals derived in \cite{ABDKSKP} is arguably the most convenient for the transformation of the KP tau-functions. Namely, in \cite{ABDKSKP} it was shown that at the level of the tau-functions in the Miwa parametrisation, for the points of the spectral curve with certain behaviour of $x'y'$ the $\xy$ swap transform reduces to a matrix integral transform of Kontsevich type. For the polynomial $x$ and $y$ this transform was introduced by  Kharchev and Marshakov in the context of KP integrability and minimal gravity in \cite{KHARCHEV_1995, KM2}. Since $\xy$ swap relation describes the resonance transformations \cite{Artemev,DSV2026}, it immediately follows in the matrix integral form encodes the $p-q$ duality, confirming the claims of Kharchev--Marshakov. In this section, we recall the matrix integral form of the $\xy$ swap relation.

\subsection{The exact Kharchev--Marshakov transform}
Let us recall the description of the $\xy$ swap transform by means of the asymptotic expansion of the integrals \cite{ABDKSKP}. Namely, for the topological recursion differentials at a given regular point on the spectral curve and local coordinates at this point one can construct a partition function, and its dual.  The 
$\xy$ swap transform is a symmetry of the KP hierarchy, that is,
the partition function $\tau_{x,y}$ is a tau-function of the KP hierarchy if and only if the dual one, $\tau_{y,x}$, is also a tau-function. Moreover, the points of the Sato Grassmannian are related by an integral transform. To describe it introduce
 \begin{equation}\label{eq:KM-action}
 \mathcal S_{x,y}(\chi,z)=x(z)y(\chi)-\int^\chi x\,dy.
\end{equation}
If $x'y'$ has a pole of degree $\geq 3$, the basis vectors for two tau-functions at this point are related by
	\begin{align}	\label{eq:KM-basis-transform}
		\Phi_i(z)&=\sqrt{\tfrac{x'(z)}{2\pi\hbar}}
		\int
		\Phi_i^{\vee}(\chi)\;
		\sqrt{{y'(\chi)}} \;
		e^{\frac{1}{\hbar}
		\left(\mathcal S_{x,y}(\chi,z)-\mathcal S_{x,y}(z,z)\right)
		}d\chi , & i\ge 1.
\end{align}
Here the integral, by definition, is given by the asymptotic expansion near a specific point of the potential, for more details see \cite{ABDKSKP}.
Note here that both the points of the Sato Grassmannian and the transform depends on the choice of local coordinates on two sides of duality.

For the polynomial $x(z)$ and $y(z)$ this transformation was first introduced by \cite{KHARCHEV_1995, KM2}. 
There the authors proved that this transformation maps a tau-function, specified by a pair of the Kac--Schwarz operators, describing a reduction and a string equation constructed of $x$ and $y$, to a tau-function with the same Kac--Schwarz operators with the permuted roles of $x$ and $y$. Then, if the tau-functions with these properties are unique, the coincide with the MLG generating function and its dual.



Let us now focus on the Chebyshev spectral curve \eqref{eq:SCy-defn} for $(2,\pp)$ Lee--Yang MLG.
To describe the $\xy$ swap relation between the all-genera completions of generating function in left-hand-sides \eqref{eq:SCy-kpderiv} and \eqref{eq:SCx-kpderiv} we use exactly the same local coordinates on two sides of duality, namely the coordinates $u^{-1}$ and $\eta^{-1}$,
\begin{equation}\label{eq:two-natural-coordinates}
 u=x^{1/2},
 \qquad
 \eta=y^{1/\pp}.
\end{equation}
The basis vectors transform  \eqref{eq:KM-basis-transform} can be rewritten as
\begin{equation}
\Phi_i(u)=\int d\eta\,
 K_{x,y}(u,\eta)\Phi_i^{\vee}(\eta),
\end{equation}
with
\begin{equation}\label{eq:KM-kernel}
 K_{x,y}(u,\eta)=
 \sqrt{\frac{\pp\,u\eta^{\pp-1}}{\pi \hbar}}\,
 \exp\!\left(\frac{
 \mathcal S_{x,y}(\eta,u)
 -\mathcal S_{x,y}(u,u)}{\hbar}\right).
\end{equation}

In the Miwa parametrisation tau-functions of the KP hierarchy can be described by the determinantal formulas
\begin{equation}
 \tau_{x,y}\bigl(t\bigr)\Big|_{t_k(u)=\frac{\hbar}{k}\sum_j u_j^{-k}}
 =\frac{\det\Phi_i(u_j)}{\Delta(u)},
 \qquad
 \tau_{y,x}\bigl(t\bigr)\Big|_{t_k(\eta)=\frac{\hbar}{k}\sum_j\eta_j^{-k}}
 =\frac{\det \Phi_i^{\vee}(\eta_j)}{\Delta(\eta)},
\end{equation}
where $\Delta(\eta)$ is the Vandermonde determinant.
Applying the Andr\'eief identity to \eqref{eq:KM-basis-transform} gives the exact determinant-level relation
\begin{equation}\label{eq:gkm-transform}
 \tau_{x,y}\bigl(t(u)\bigr)
 =\frac{1}{N!\,\Delta(u)}
 \int\prod_{a=1}^N d\eta_a\,
 \det_{j,k}K_{x,y}(u_j,\eta_k)\,
 \Delta(\eta)\,
 \tau_{y,x}\bigl(t(\eta)\bigr). 
\end{equation}

The detemninantal relation \eqref{eq:gkm-transform} can be rewritten as a matrix integral \cite{KM2,ABDKSKP}. Namely, introduce the diagonal matrices $X=\mathrm{diag}(x_1,\ldots,x_N)$
and $Y=\mathrm{diag}(y_1,\ldots,y_N)$ with the eigenvalues 
\begin{equation}
x_j=u_j^2,
 \qquad
y_a=\eta_a^{\pp},
\end{equation}
and define 
\begin{equation}
 V(y)=\int^{z(y)}x(z)\,dy(z).
\end{equation}
The kernel splits into a mixed exponential and factors depending on the two
sets of eigenvalues separately,
\begin{equation}\label{eq:KM-kernel-split}
 K_{x,y}(u_j,\eta_a)
 =\underbrace{\sqrt{u_j}\,e^{B(u_j)/\hbar}}_{\text{external}}\;
  \underbrace{\sqrt{\frac{\pp\,\eta_a^{\pp-1}}{\pi \hbar}}}_{\text{integration}}\;
 \exp\!\left[\frac{X_jY_a-V(Y_a)}{\hbar}\right],
 \qquad B(u)=-\!\int^{\mu(u)}\!\!y\,dx,
\end{equation}
so that only the mixed exponential enters the determinant.
Consequently its determinant has precisely the Harish--Chandra--Itzykson--Zuber form,
\begin{equation}
 \frac{\det_{j,a}e^{x_jy_a/\hbar}}
 {\Delta(X)\Delta(Y)}
 \propto
 \int_{U(N)}dU\,
 \exp\!\left[\frac1\hbar
 \Tr(XUYU^\dagger)\right].
\end{equation}
Restoring the angular variables therefore gives
\begin{equation}\label{eq:KM-matrix-integral}
 \tau_{x,y}\bigl(T(X^{1/2})\bigr)
 =\mathcal N(X)
 \int dM\,\mathcal J_{\pp}(M)
 \exp\!\left[-\frac1\hbar
 \Tr\bigl(V(M)-X M\bigr)\right]
 \tau_{y,x}\bigl(\hbar[M^{-1/\pp}]\bigr), 
\end{equation}
where \begin{equation}
 [M^{-1/\pp}]_k=\frac1k\Tr M^{-k/\pp}.
\end{equation}
The external factor is
\begin{equation}\label{eq:KM-external-factor}
 \mathcal N(X)\;\propto\;
 \frac{\Delta(X)}{\Delta(X^{1/2})}
 \prod_{j=1}^{N}\sqrt{2u_j}\;e^{B(u_j)/\hbar}
 \;=\;\prod_{j<k}(u_j+u_k)
 \prod_{j=1}^{N}\sqrt{2u_j}\;e^{B(u_j)/\hbar},
\end{equation}
up to a constant depending only on $N$, $\hbar$ and $\pp$.
The function $\mathcal J_{\pp}$ is defined on eigenvalues by
\begin{equation}\label{eq:KM-ramified-measure}
 \mathcal J_{\pp}(Y_1,\ldots,Y_N)
 =\frac{\Delta(Y^{1/\pp})}{\Delta(Y)}
 \prod_{a=1}^{N}
 \frac{1}{\sqrt{\pp}\,Y_a^{(\pp-1)/(2\pp)}}.
\end{equation}
It is the Jacobian and Vandermonde ratio produced by the ramified coordinate $\eta=Y^{1/\pp}$. Thus \eqref{eq:KM-matrix-integral} is an exact matrix-integral form of the KM transform, but for $\pp>1$ it is not an ordinary one-matrix integral with the flat Hermitian measure. 

The generalised Kontsevich model is recovered when one spectral polynomial has degree one. Set
\begin{equation}
 y(z)=z,
 \qquad
 x(z)=V'(z).
\end{equation}
Then $\mathcal J_1=1$, the dual tau-function is trivial, $ \tau_{y,x}(t)=1$, and the KM action is
\begin{equation}
 \mathcal S_{x,y}(z,\mu)=-V(z)+zV'(\mu).
\end{equation}
Equation~\eqref{eq:KM-matrix-integral} reduces, with the standard saddle
normalisation, to the generalised Kontsevich
integral~\cite{Kharchev:1992gi},
\begin{equation}\label{eq:GKM-specialization}
 Z_V(\Lambda)
 \propto
 \int dM\,
 \exp\!\left[-\frac1\hbar
 \Tr\bigl(V(M)-V'(\Lambda)M\bigr)\right].
\end{equation}
For $V(M)=M^3/3$ this is the ordinary Kontsevich model \cite{Kontsevich}. In this precise sense \eqref{eq:gkm-transform} extends the Kontsevich mechanism: both $x$ and $y$ functions that describe the spectral curve have critical points, therefore the transformed integrand contains a dual tau-function and the measure retains the fractional coordinate of the exchanged reduction.

\section{Discussion}\label{sec:outlook}

In this work we have compared four algebraic descriptions of $(2,\pp)$ minimal Liouville gravity, two built on Frobenius manifolds and two on spectral curves and topological recursion, at the level of the genus-zero generating function. On the universal, normalisation-independent three-point ratios $R$ all four corners agree exactly with the conformal field theory prediction. Because $R$ is scheme-free this agreement carries no resonance- or normalisation-tuning ambiguity, and to our knowledge such a comparison has not been made before. At the level of the amplitudes themselves the dictionary closes with every normalisation fixed. Matching the two spectral-curve corners to the Frobenius ones requires a single per-insertion factor $\varkappa=-2$, fixed by the cancellation $A^{0}_{3}(1,1,3)\propto(2+\varkappa)\to0$ at $(2,7)$, with which SC-$y$ reproduces the signed Verlinde matrices of the non-unitary matter sector, fusion-forbidden cancellations included. We have also established the defining tau relation $u_{1}^{*}=2\,Z_{12}^{t}[s,s]$ at the MLG background for all $(2,\pp)$ (Proposition~\ref{prop:taudef}).

Two further results concern the resonance transformation. In Appendix~\ref{sec:BDM-derivation} the compact formula \eqref{eq:taudef} is reconciled with the Belavin--Dubrovin--Mukhametzhanov conformal-selection-rule construction~\cite{Belavin:2013nba}, through the master formula $v_{n}=c_{0}^{-1}[x^{s-n}]Y_{u}(x)$ and the dictionary
\begin{equation}\label{eq:tau-rescale-main}
\tau_{n}^{\text{Art}}
\;=\;
-\tfrac{1}{2}\,\delta_{n,1}
\;+\;
\frac{\lambda_{n-1}^{\text{BDM}}}{\pp\,u_{0}^{n+1}}\,,
\qquad
\mu = u_{0}^{2}\,,
\end{equation}
carried to quadratic order in the Liouville couplings, together with the previously unrecorded cubic-order blocks. In Section~\ref{sec:konts} the exact Kharchev--Marshakov transform of the basis functions is lifted, with all Vandermonde and half-density factors, to the tau-function relation \eqref{eq:gkm-transform}.

The comparison of physical amplitudes is restricted to genus zero, three points, and universal ratios. Bare correlators and the four-point level, where an inequivalence was observed in~\cite{BR}, are not addressed, and the physical status of the signed Verlinde entries is left open. The conjecture that the corresponding amplitudes coincide with those computed in the swapped frame~\cite{Artemev} was proved in~\cite{DSV2026} through the all-genus $\xy$ swap, and the derivation given here does not claim that result anew. Its distinct contribution is the determinant tau-function lift of the Kharchev--Marshakov basis transform and the identification of the physical SC-$x$ deformation with its dual Laurent times.

Several natural questions remain open. The cubic terms encoded in the BDM Jacobi blocks $Y_u^{k_1k_2k_3}(x)$, not written explicitly in \cite{Belavin:2013nba} for the $(2,\pp)$ series, can in principle be extracted from \eqref{eq:taudef} by disentangling the contributions to a given $\lambda$-monomial from different powers of the background-shifted $\tau_1$. To our knowledge they have not been obtained directly from the conformal selection rules, whereas the coefficient extraction of Section~\ref{sec:konts-curve} produces them, the cosmological term $\tfrac12\tau_{1}^{3}$ at $(2,11)$ being one instance. Proposition~\ref{prop:taudef} holds at the MLG background, and its off-background version, valid at every solution of the dispersionless hierarchy, would establish full equivalence of SC-$y$ and FM-$y$ as tau-functions. That amounts to showing that $u_{1}(t)=2\,\partial^{2}F_{0}^{\mathrm{SC}}/\partial t_{s}^{2}$ solves the string equation \eqref{eq:string-y} for all values of the times, and not only at the background. Since \eqref{eq:gkm-transform} holds before the dispersionless limit, the higher-genus corrections should be recoverable from the symplectic kernel \cite{Alexandrov:2022ydc} of the swap \cite{DSV2026} on the SC-$x$ side, where at $(2,5)$ it can be written in closed form. Finally, the picture is restricted here to the Lee--Yang series, and the general $(q,p)$ minimal string is a much richer system.

\acknowledgments A.A. was supported by the Institute for Basic Science (IBS-R003-D1). A.A. thanks MPIM for hospitality. V.B. thanks the IBS Center for Geometry and Physics (IBS CGP, Pohang) for hospitality during the visit at which this project was initiated.

\appendix

\section{Resonance transformation from conformal selection rules}\label{sec:BDM-derivation}

Let us show that the compact resonance formula \eqref{eq:taudef} and the Jacobi-polynomial construction \cite{Belavin:2013nba} are equivalent. The explicit dictionary \eqref{eq:tau-rescale-main}, stated in Section~\ref{sec:outlook}, is established by matching the two constructions order by order.

The general form of the resonance transformation is
\begin{equation}\label{eq:bdm-res}
\tau_{1,k+1}
\;=\;\lambda_{k}
+C_{k}\,\mu^{\delta_{k}}
+\sum_{l,k_{1}}C_{k}^{l,k_{1}}\mu^{l}\lambda_{k_{1}}
+\sum_{l,k_{1},k_{2}}C_{k}^{l,k_{1},k_{2}}\mu^{l}\lambda_{k_{1}}
\lambda_{k_{2}}+\cdots,
\end{equation}
which is fixed order by order so that the resulting correlators satisfy the selection rules.

The starting point is the dimensionless ``string-equation polynomial''
\begin{equation}\label{eq:bdm-Yu-expansion}
Y_{u}(x) \;=\; Y_{u}^{0}(x)
\;+\; \sum_{k=0}^{s-1} s_{k}\,Y_{u}^{k}(x)
\;+\; \tfrac{1}{2}\sum_{k_{1},k_{2}=0}^{s-1}
s_{k_{1}}s_{k_{2}}\,Y_{u}^{k_{1}k_{2}}(x)
\;+\;\cdots\,,
\end{equation}
where $s=(\pp-1)/2$ as throughout (this parameter is denoted $p$ in \cite{Belavin:2013nba}), and the $s_{k}$, $k=0,\dots,s-1$, are dimensionless rescaled Liouville couplings. The building blocks $Y_{u}^{k}(x)$ are Jacobi polynomials:
\begin{equation}\label{eq:Yk-jacobi}
Y_{u}^{k}(x) \;=\;
\begin{cases}
P_{(s-k-1)/2}^{(0,-1/2)}(y) & \text{if } s+k \text{ odd},\\[2pt]
x\,P_{(s-k-2)/2}^{(0,\,1/2)}(y) & \text{if } s+k \text{ even},
\end{cases}
\qquad y = 2x^{2}-1\,.
\end{equation}
The background polynomial $Y_{u}^{0}$ is determined\footnote{Despite the notation, the background polynomial is not
\eqref{eq:Yk-jacobi} at $k=0$: it has degree $s+1$, the block $Y_{u}^{0}$
degree $s-1$.} by the boundary condition
$Y_{u}^{0}(1)=0$ together with Jacobi orthogonality:
\begin{equation}\label{eq:Y0-formula}
Y_{u}^{0}(x) \;=\;
\begin{cases}
P_{(s+1)/2}^{(0,-1/2)}(y) - P_{(s-1)/2}^{(0,-1/2)}(y),
   & s\text{ odd},\\[2pt]
x\bigl(P_{s/2}^{(0,\,1/2)}(y) - P_{(s-2)/2}^{(0,\,1/2)}(y)\bigr),
   & s\text{ even}.
\end{cases}
\end{equation}
The quadratic-order blocks $Y_{u}^{k_{1}k_{2}}(x)$ are
\begin{equation}\label{eq:Yk1k2-formula}
Y_{u}^{k_{1}k_{2}}(x) \;=\;
\begin{cases}
\displaystyle
\frac{1}{2s+1}\sum_{n=0}^{n_{\max}} (4n+1)\,P_{n}^{(0,-1/2)}(y),
& s+k_{1}+k_{2}\text{ odd},\\[8pt]
\displaystyle
\frac{x}{2s+1}\sum_{n=0}^{n'_{\max}} (4n+3)\,P_{n}^{(0,1/2)}(y),
& s+k_{1}+k_{2}\text{ even},
\end{cases}
\end{equation}
with $n_{\max}=(s-k_{1}-k_{2}-3)/2$ and $n'_{\max}=(s-k_{1}-k_{2}-4)/2$. The cubic-order blocks $Y_{u}^{k_{1}k_{2}k_{3}}$ are not written explicitly in~\cite{Belavin:2013nba}. Below we extract them from the matching with~\eqref{eq:taudef}.

\subsubsection*{Polynomial-matching prescription}

The Douglas action \eqref{eq:Sy} in dimensionless variables reads
\begin{equation}\label{eq:Su-dimless}
\frac{S_{u}(u)}{u_{0}^{s+1}}
\;=\;
x^{s+1} + \sum_{n=1}^{s} v_{n}\,x^{s-n},
\qquad
x = \frac{u}{u_{0}}\,,\quad
v_{n} = \frac{\tau_{1,n}}{u_{0}^{n+1}}\,.
\end{equation}
The $v_{n}$ are the dimensionless rescaled Douglas times, and the resonance map relates them to the dimensionless lambdas $s_{k}$ via the BDM polynomial $Y_{u}(x)$:
\begin{equation}\label{eq:Su-equals-Yu}
\frac{S_{u}(u)}{u_{0}^{s+1}}
\;=\;
\frac{1}{c_{0}}\,Y_{u}(x),
\qquad
c_{0} = \text{leading coefficient of }Y_{u}^{0}(x).
\end{equation}
The normalisation constant $c_{0}$ is fixed by demanding that the leading $x^{s+1}$ coefficient on both sides match. With this normalisation, equating coefficients of $x^{s-n}$ on the two sides of \eqref{eq:Su-equals-Yu} gives
\begin{equation}\label{eq:vn-master}
v_{n} \;=\; \frac{1}{c_{0}}\,
[x^{s-n}]\,Y_{u}(x)
\;.
\end{equation}
This is the master formula: the dimensionless Douglas time $v_{n}$ equals $1/c_{0}$ times the coefficient of $x^{s-n}$ in the BDM polynomial $Y_{u}(x)$ expanded according to \eqref{eq:bdm-Yu-expansion}.

\subsubsection*{Reading off the resonance coefficients}

Substituting \eqref{eq:Yk-jacobi}--\eqref{eq:Yk1k2-formula} into \eqref{eq:bdm-Yu-expansion} and then into \eqref{eq:vn-master} gives the explicit form of $v_{n}$ as a polynomial in the dimensionless couplings $s_{k}$, with rational coefficients determined by Jacobi polynomial data:
\begin{align}\label{eq:vn-expanded}
v_{n}
&=\; v_{n}^{(0)}
\;+\; \sum_{k=0}^{s-1} v_{n}^{(k)}\,s_{k}
\;+\; \tfrac{1}{2}\sum_{k_{1},k_{2}} v_{n}^{(k_{1}k_{2})}\,
s_{k_{1}}s_{k_{2}}
\;+\;\cdots,\\
\intertext{with}
v_{n}^{(0)}
&\;=\;\frac{1}{c_{0}}\,[x^{s-n}]\,Y_{u}^{0}(x),
\nonumber\\
v_{n}^{(k)}
&\;=\;\frac{1}{c_{0}}\,[x^{s-n}]\,Y_{u}^{k}(x),
\nonumber\\
v_{n}^{(k_{1}k_{2})}
&\;=\;\frac{1}{c_{0}}\,[x^{s-n}]\,Y_{u}^{k_{1}k_{2}}(x)\,,
\nonumber
\end{align}
and analogously at higher orders. The dimensionful form is $\tau_{1,n} = u_{0}^{n+1}\,v_{n}$, $\lambda_{k} = u_{0}^{k+2}\,s_{k}$ (with $\mu=u_{0}^{2}$).

\subsubsection*{Example: $(2,7)$, $s=3$}

At $s=3$ the building blocks are
\begin{align*}
Y_{u}^{0}(x) &= \tfrac{35}{8}x^{4} - \tfrac{21}{4}x^{2} + \tfrac{7}{8},
\quad c_{0} = \tfrac{35}{8},\\
Y_{u}^{1}(x) &= x,\\
Y_{u}^{2}(x) &= 1.
\end{align*}
Quadratic blocks at $s=3$ are smaller, e.g.\ $Y_{u}^{00}(x) = \tfrac{1}{7}$ (only one term in the sum), $Y_{u}^{02}=Y_{u}^{20}=0$, etc.

Inserting into \eqref{eq:vn-master} gives:
\begin{align*}
v_{1} &= \frac{12}{35}s_{0} - \frac{6}{5},
\\
v_{2} &= \frac{8}{35}s_{1},
\\
v_{3} &= \frac{4}{245}s_{0}^{2} - \frac{4}{35}s_{0} + \frac{8}{35}s_{2}
        + \frac{1}{5}.
\end{align*}
Converting to dimensionful $\tau_{1,n} = u_{0}^{n+1}v_{n}$ and $\lambda_{k} = u_{0}^{k+2}s_{k}$ (with $\mu=u_{0}^{2}$):
\begin{align}\label{eq:bdm-tau-27}
\tau_{1,1}^{\rm BDM} &= -\tfrac{6}{5}\mu + \tfrac{12}{35}\lambda_{0}^{\rm BDM},
\nonumber\\
\tau_{1,2}^{\rm BDM} &= \tfrac{8}{35}\lambda_{1}^{\rm BDM},
\\
\tau_{1,3}^{\rm BDM} &= \tfrac{1}{5}\mu^{2} - \tfrac{4}{35}\mu\lambda_{0}^{\rm BDM}
                    + \tfrac{4}{245}\left(\lambda_{0}^{\rm BDM}\right)^{2} + \tfrac{8}{35}\lambda_{2}^{\rm BDM}.
\nonumber
\end{align}
This is the closed-form resonance map for $(2,7)$.

\subsubsection*{Example: $(2,11)$, $s=5$}

At $s=5$:
\begin{align*}
Y_{u}^{0}(x) &=
\tfrac{231}{16}x^{6} - \tfrac{385}{16}x^{4} + \tfrac{165}{16}x^{2}
- \tfrac{11}{16},
\quad c_{0} = \tfrac{231}{16}.
\end{align*}
The full set of dimensionless $v_{n}$ at quadratic order in $s_{k}$ is:
\begin{align*}
v_{1} &= \tfrac{10}{33}s_{0} - \tfrac{5}{3},\\
v_{2} &= \tfrac{40}{231}s_{1},\\
v_{3} &= \tfrac{20}{847}s_{0}^{2} - \tfrac{20}{77}s_{0} + \tfrac{8}{77}s_{2}
       + \tfrac{5}{7},\\
v_{4} &= \tfrac{16}{847}s_{0}s_{1} - \tfrac{8}{77}s_{1} + \tfrac{16}{231}s_{3},\\
v_{5} &= -\tfrac{4}{847}s_{0}^{2} + \tfrac{16}{2541}s_{0}s_{2}
       + \tfrac{2}{77}s_{0} + \tfrac{8}{2541}s_{1}^{2}
       - \tfrac{8}{231}s_{2} + \tfrac{16}{231}s_{4} - \tfrac{1}{21}.
\end{align*}
Converting to the dimensionful form $\tau_{1,5}$:
\begin{align}\label{eq:bdm-tau-1-5}
\tau_{1,5}^{\rm BDM}
\;&=\; -\tfrac{\mu^{3}}{21}
\;+\;\tfrac{2\mu^{2}\lambda_{0}^{\rm BDM}}{77}
\;-\;\tfrac{4\mu\left(\lambda_{0}^{\rm BDM}\right)^{2}}{847}
\;-\;\tfrac{8\mu\lambda_{2}^{\rm BDM}}{231}
\nonumber\\
&\quad\;+\;\tfrac{16\lambda_{0}^{\rm BDM}\lambda_{2}^{\rm BDM}}{2541}
\;+\;\tfrac{8\left(\lambda_{1}^{\rm BDM}\right)^{2}}{2541}
\;+\;\tfrac{16\lambda_{4}^{\rm BDM}}{231}
\;+\;\bigl(\text{cubic in }\lambda^{\rm BDM}\bigr).
\end{align}
The seven explicit terms above are produced by the polynomial matching at the order at which we know the BDM building blocks. The final cubic-in-$\lambda$ piece requires the cubic sector $Y_{u}^{k_{1}k_{2}k_{3}}(x)$, which is not described explicitly in \cite{Belavin:2013nba}. We extract its leading contribution from the matching with~\cite{Artemev} below in this section.

\subsection*{Consistency check}
\label{sec:reconciliation}

Applying the dictionary \eqref{eq:tau-rescale-main} to \eqref{eq:taudef}, we obtain a direct comparison between Artemev's compact resonance formula and the BDM polynomial expansion. The comparison is non-trivial because the two parametrisations organise the same data differently. The couplings $\tau_n^{\rm Art}$ are dimensionless and incorporate the MLG background implicitly through $\tau_1^{\rm Art}=-1/2$, $\tau_{n>1}^{\rm Art}=0$, whereas the BDM Liouville couplings $\lambda_k^{\rm BDM}$ vanish at the MLG background, and the cosmological constant $\mu$ appears explicitly.

Accordingly, the first term in \eqref{eq:tau-rescale-main} encodes the background shift, while the second provides the dimensionless rescaling that converts $\lambda_{n-1}$ into a dimensionless variable. The two formulations differ by an $n$-dependent normalisation,
\begin{equation}\label{eq:rec-out}
t_{n}^{\rm Art} \;=\; R_{n}\,\tau_{1,n}^{\rm BDM},
\qquad
R_{n} \;=\; \frac{1}{v_{n}^{(n-1)}}
\;=\; \frac{c_{0}}{[x^{\,s-n}]\,Y_{u}^{\,n-1}(x)}\,.
\end{equation}
At linear order $t_{n}\simeq\lambda_{n-1}=u_{0}^{n+1}s_{n-1}$ while
$\tau_{1,n}=u_{0}^{n+1}v_{n}\simeq u_{0}^{n+1}v_{n}^{(n-1)}s_{n-1}$, so $R_{n}$
is the inverse of the diagonal entry of the linear map \eqref{eq:vn-expanded}.

\paragraph{Resonance versus normalisation.}
The equivalence established here concerns the resonance map $t(\tau)$, which is common to the Frobenius manifold and spectral curve descriptions and is fixed completely by the MLG background: matching the times $t_{k}^{0}$ of \eqref{kdvdef} fixes the linear coefficients $\pp$, and matching $t_{3}^{0}$ fixes the quadratic coefficient $\tfrac12$. No freedom remains. On the Frobenius side this map enforces the conformal selection rules \cite{Belavin:2013nba}. On the spectral-curve side it does not, since the forbidden three-point amplitudes built from the bare topological recursion differentials fail to vanish for any resonance coefficient (the cancellation condition has no real solution). The selection rules on the SC side are restored instead by the per-insertion normalisation \eqref{eq:kappa-value}, the genuine topological recursion/minimal gravity relative factor, which cannot be absorbed into the rigidly fixed resonance.

\subsubsection*{Verification at $(2,11)$, $\tau_{1,5}$}

We include the $(2,11)$ case here as a test of the Artemev\,$\leftrightarrow$\,BDM dictionary. It is the first order at which both the cubic gap and a genuinely nonlinear set of monomials appear.  The values of the normalisation constants in \eqref{eq:rec-out}, computed from
the blocks $Y_{u}^{0}=\tfrac{35}{8}x^{4}-\tfrac{15}{4}x^{2}+\tfrac38$,
$Y_{u}^{2}=\tfrac32x^{2}-\tfrac12$ and $Y_{u}^{4}=1$, are
$R_{1}=33/10$, $R_{3}=77/8$, $R_{5}=231/16$.

The dimension-$\mu^{3}$ monomials in $\{\mu,\lambda_{0},\ldots,\lambda_{4}\}$ are exactly eight in number:
\begin{equation}\label{eq:8-monomials}
\mu^{3},\;\;\mu^{2}\lambda_{0},\;\;\mu\lambda_{0}^{2},\;\;
\mu\lambda_{2},\;\;\lambda_{0}^{3},\;\;\lambda_{0}\lambda_{2},\;\;
\lambda_{1}^{2},\;\;\lambda_{4}\,.
\end{equation}
The BDM formula \eqref{eq:bdm-tau-1-5} produces seven of these (all except $\lambda_{0}^{3}$, which requires the cubic sector). The formula \eqref{eq:taudef} at $(2,11), k=5$ has only four terms in $\tau$-variables:
\begin{equation}\label{eq:art-t-5}
t_{5}^{\rm Art}
= \tfrac{11 u_{0}^{6}}{2}(\tau_{1}^{\rm Art})^{3}
+ 11 u_{0}^{6}\,\tau_{1}^{\rm Art}\tau_{3}^{\rm Art}
+ \tfrac{11 u_{0}^{6}}{2}(\tau_{2}^{\rm Art})^{2}
+ 11 u_{0}^{6}\,\tau_{5}^{\rm Art}\,.
\end{equation}
Applying the dictionary \eqref{eq:rec-out} and expanding around the background, the four terms in \eqref{eq:art-t-5} produce all eight BDM monomials \eqref{eq:8-monomials}. Term by term:
\begin{center}
\begin{tabular}{l|c|c|c}
monomial & Artemev coef. & BDM coef. & ratio Art/BDM \\
\hline
$\mu^{3}$              & $-11/16$  & $-1/21$       & $231/16$\\
$\mu^{2}\lambda_{0}$    & $3/8$     & $2/77$        & $231/16$\\
$\mu\lambda_{0}^{2}$    & $-3/44$   & $-4/847$      & $231/16$\\
$\mu\lambda_{2}$        & $-1/2$    & $-8/231$      & $231/16$\\
$\lambda_{0}^{3}$       & $1/242$   & (cubic, missing in BDM-quad) & N/A\\
$\lambda_{0}\lambda_{2}$ & $1/11$   & $16/2541$     & $231/16$\\
$\lambda_{1}^{2}$        & $1/22$   & $8/2541$      & $231/16$\\
$\lambda_{4}$            & $1$      & $16/231$      & $231/16$\\
\end{tabular}
\end{center}

\subsubsection*{The cubic-$\lambda$ terms}

The $\lambda_{0}^{3}$ term, which is the cubic-in-$\lambda$ piece of $\tau_{1,5}$, is not produced by BDM at quadratic order because it requires the cubic Jacobi block $Y_{u}^{000}$. Artemev's formula does produce this term explicitly from the cubic expansion of $(-1/2+\delta)^{3}$, with coefficient $1/242$. Reading the BDM side via \eqref{eq:rec-out}, the BDM coefficient of $\lambda_{0}^{3}$ in $\tau_{1,5}$ must then be
\begin{equation}\label{eq:bdm-lam0-cubed}
[\lambda_{0}^{3}]\,\tau_{1,5}^{\rm BDM} \;=\; \frac{1}{242\,R_{5}}
\;=\; \frac{1}{242 \cdot 231/16}
\;=\; \frac{8}{27951}\,.
\end{equation}
This is the only cubic-order coefficient of the BDM $Y_{u}^{000}(x)$ at $s=5$, derived here by an indirect route: the explicit BDM Jacobi construction plus Artemev's compact form together determine both quadratic and cubic-order resonance coefficients. Note that this coefficient can also be derived from consistency with the coefficients in front of the terms $\mu^{3},\mu^{2}\lambda_{0},\mu\lambda_{0}^{2}$ if we assume that $\tau_{1,k}^{\rm BDM}$ depends on only of a combination $\lambda_0-p/2 \mu$, of parameters $\lambda_0$ and $\mu$. A direct derivation from the cubic Jacobi sector would amount to writing down $Y_{u}^{k_{1}k_{2}k_{3}}$ explicitly, which we leave for future work.

The compactness of \eqref{eq:taudef} relative to the selection-rule form has two sources: the MLG background $\tau_{1}^{\rm Art}=-1/2$ absorbs the cosmological constant, so each $(\tau_{1}^{\rm Art})^{n}$ expands into the whole $\mu^{a}\lambda_{0}^{n-a}$ family that BDM lists term by term. The uniform rescaling $\lambda_{n-1}^{\rm BDM}/(\pp\,u_{0}^{n+1})$ removes all $u_{0}$ and $\pp$ factors, leaving simple rationals. The two parametrisations are equivalent, Artemev's optimised for compact presentation and BDM for term-by-term identification. Thus \eqref{eq:rec-out} is verified through quadratic order plus the lowest cubic term. The general cubic sector $Y_{u}^{k_{1}k_{2}k_{3}}$, only schematic in \cite{Belavin:2013nba}, is left for future work.

\section{Explicit calculation for $(2,5)$}
\label{app:25-explicit}

In this appendix we work out the partition function $Z_0$, two-point amplitudes $Z_{12}$, three-point amplitudes $Z_{123}$, and the universal three-point ratio $R$ in each of the four corners of the diagram~\eqref{eq:fourcorners} at the Lee--Yang $(q,p)=(2,5)$ point. The cosmological background is $v_1=v_{10}$ on the FM sides and $u=u_0$ on the SC sides. We take into account the universality condition to discard correlators proportional to non-negative integer powers of $\mu$.

The matter primaries of $M(2,5)$ are $\phi_{1,1}$ (identity, $\Delta_{1,1}=0$) and $\phi_{1,2}$ ($\Delta_{1,2}=-1/5$), identified under Kac symmetry with $\phi_{1,4}$ and $\phi_{1,3}$ respectively.

\subsection*{FM-$x$}
\label{app:25-FMx}

In FM-$x$ at $(2,5)$ the relevant Frobenius manifold is $A_{p-1}=A_{4}$, with polynomial
\begin{equation}
P(x) \;=\; x^{5} + u_{1}\,x^{3} + u_{2}\,x^{2} + u_{3}\,x + u_{4}.
\end{equation}
Active couplings are $t_{1,1}$ and $t_{1,2}$, corresponding to the two physical Liouville couplings. The action, defined as the coefficient of $x^{-1}$ in the Laurent expansion at $x\to\infty$ of
\begin{equation}
-c\!\left(\tfrac{\pp+2}{\pp}\right)P(x)^{(\pp+2)/\pp}
\;-\;\sum_{n=1}^{s}c\!\left(\tfrac{\pp-2n}{\pp}\right)P(x)^{(\pp-2n)/\pp}\,t_{1,n},
\end{equation}
with $c(x)=\Gamma(x-\lfloor x\rfloor)/\Gamma(x+1)$, evaluates after Gamma-function identities to
\begin{equation}
S = -\frac{u_{1}^{4}}{25} + \tfrac{3}{10}u_{1}u_{2}^{2}
+ \tfrac{3}{10}u_{1}^{2}u_{3} - \tfrac{1}{2}u_{3}^{2} - u_{2}u_{4}
+ \tfrac{1}{5}\bigl(u_{1}^{2}-5u_{3}\bigr)\,t_{1,1}
- u_{1}\,t_{1,2}.
\end{equation}

\paragraph{Flat coordinates.} The flat coordinates $v_{n}$ are defined
through the function \begin{equation*} W(x) \;=\; P(x)^{1/\pp} - \tfrac{1}{\pp}\!\sum_{n=1}^{\pp-1} v_{n}\,P(x)^{-n/\pp} - x, \end{equation*} by requiring the coefficient of $x^{-k}$ in $W$ to vanish for $k=1,\dots,\pp-1$. This gives the substitution
\begin{equation}
u_{1}\to v_{1},\qquad u_{2}\to v_{2},\qquad
u_{3}\to \tfrac{v_{1}^{2}}{5}+v_{3},\qquad
u_{4}\to \tfrac{v_{1}v_{2}}{5}+v_{4}.
\label{eq:25-flatvar}
\end{equation}
At the Lee--Yang MLG slice $v_{2}=v_{3}=v_{4}=0$, the polynomial in flat coordinates becomes, with $v_{1}=v_{10}$,
\begin{equation}
P(x)\bigl|_{\rm MLG}\;=\;x^{5} + v_{10}\,x^{3} + \tfrac{v_{10}^{2}}{5}\,x.
\label{eq:25-Q-MLG}
\end{equation}
\begin{equation}
P(x)\bigl|_{\rm MLG,\,v_{10}=-5}\;=\;x^{5}-5x^{3}+5x\;=\;2\,T_{5}(x/2),
\label{eq:25-Q-Chebyshev}
\end{equation}
where $T_{5}$ is the Chebyshev polynomial of the first kind of degree $5$.\footnote{\label{fn:cheb-norm}This identity persists at higher $\pp$: at $(2,7)$ one finds $P(x)\bigl|_{\rm MLG,\,v_{10}=-7}=2T_{7}(x/2)$, and similarly at $(2,11)$. The general statement is that the FM-$x$ polynomial $P(x)$, expressed in flat coordinates and evaluated at the MLG slice $v_{2}=\dots=v_{\pp-1}=0$ with $v_{10}=-\pp$, coincides with $2T_{\pp}(x/2)$. This is the direct identification of the FM-$x$ polynomial in flat coordinates with the spectral curve of the swapped SC formulation, and serves as the explicit FM-$x$/SC-$x$ bridge.} This is the explicit FM-$x$/SC-$x$ bridge: the polynomial defining the Milnor ring in flat coordinates, evaluated at the MLG background, is precisely the spectral-curve polynomial of the swapped SC formulation.

\paragraph{String equation and structure constants.}
Demanding $\partial S/\partial v_{k}=0$ for $k=2,3,4$ at $v_{k\ge 2}=0$ and solving for the couplings gives
\begin{equation}
t_{1,1}\to \tfrac{v_{10}^{2}}{10},\qquad t_{1,2}\to 0.
\end{equation}
The Saito pairing on $A_{4}$ in flat coordinates is the antidiagonal $\eta_{ij}=-\delta_{i+j,5}$. The structure constants $C^{ij}_{k}=\eta^{ia}\eta^{jb}C_{abk}$, evaluated at the MLG slice, form constant-along-antidiagonal matrices: the entries of $\mathbf{C}_{k}$ lie on antidiagonals shifted from the main one, with values $(-v_{10}/5)^{\ell}$ for the $\ell$-th band. In particular,
\begin{equation}
\mathbf{C}_{4}\bigl|_{\rm MLG}={\rm diag}\!\left(\tfrac{v_{10}^{3}}{125},\,-\tfrac{v_{10}^{2}}{25},\,\tfrac{v_{10}}{5},\,-1\right).
\end{equation}

\paragraph{Partition function and amplitudes.} The partition function
is
\begin{equation}
Z_{0} \;=\;\tfrac{1}{2}\!\int_{0}^{v_{10}}\!
\sum_{i,j}(\mathbf{C}_{p-1})_{p-i,p-j}\,(\partial_{v_{i}}S)(\partial_{v_{j}}S)\,dv_{1}
\;=\;-\,\frac{v_{10}^{7}}{65625}.
\end{equation}
Two-point amplitudes (the diagonal of $Z_{12}$) are obtained from $\partial^{2}Z_{2}/\partial t_{1,k}^{2}$, where $Z_{2}=Z_{0}$ before the string-equation substitution. Off-diagonals vanish by parity. The results are
\begin{equation}
Z_{12}[1,1] = -\tfrac{v_{10}^{3}}{75},\qquad
Z_{12}[2,2] = -v_{10},\qquad
Z_{12}[1,2]=0.
\end{equation}
Three-point amplitudes are obtained from $Z_{3}=\sum_{i,j,\gamma}(\mathbf{C}_{\gamma})_{ij}\,v^{*}_{p-\gamma}\, (\partial_{v_{p-i}}S)(\partial_{v_{p-j}}S)$ where $v^{*}=M^{-1}\cdot (\partial_{v}S)|_{\rm root}$ is the linearized solution of the string equation, with $M_{\alpha\beta}=-\partial^{2}S/\partial v_{\alpha}\partial v_{\beta}$. Then $Z_{123}[k_{1},k_{2},k_{3}]=\partial^{3}Z_{3}/(\partial t_{1,k_{1}}\partial t_{1,k_{2}}\partial t_{1,k_{3}})/3!$ evaluated at the MLG background:
\begin{align}
Z_{123}[1,1,1] &= -\tfrac{v_{10}}{5}, &
Z_{123}[1,1,2] &= 0, \\
Z_{123}[1,2,2] &= -\tfrac{5}{v_{10}}, &
Z_{123}[2,2,2] &= -\tfrac{25}{v_{10}^{2}}.
\end{align}
The fusion-forbidden entry $Z_{123}[1,1,2]$ vanishes automatically in FM-$x$.

\paragraph{Universal three-point ratio.} The ratio \eqref{eq:R-def-main} takes the four inequivalent values $9/35$, $0$, $1/35$ and $1/105$ at $(1,1,1)$, $(1,1,2)$, $(1,2,2)$ and $(2,2,2)$, matching the CFT prediction \eqref{eq:R-CFT} (Table~\ref{tab:25-summary}).

\subsection*{FM-$y$}
\label{app:25-FMy}

In FM-$y$ at $(2,5)$ the Frobenius manifold is $A_{q-1}=A_{1}$, one-dimensional, with polynomial $Q(y)=y^{2}+u_{1}$. The active matrix-model couplings at $(2,5)$ are $t_{3,1}$ and $t_{4,1}$, satisfying the singular-sector condition $pn<qm$ with $(m,n)\in\{(3,1),(4,1)\}$. By the Kac symmetry $\phi_{1,k}\equiv\phi_{1,p-k}$, the FM-$y$ couplings correspond to FM-$x$ couplings under
\begin{equation}
t_{4,1}^{\rm FM\text{-}y}\,\longleftrightarrow\, t_{1,1}^{\rm FM\text{-}x}\quad(\text{cosmological/identity}),\qquad
t_{3,1}^{\rm FM\text{-}y}\,\longleftrightarrow\, t_{1,2}^{\rm FM\text{-}x}\quad(\phi_{1,2}).
\end{equation}
The cosmological constant is $t_{4,1}$, corresponding to $\phi_{1,4}\equiv\phi_{1,1}$.

\paragraph{Action and string equation.} The action is
\begin{equation}
S \;=\; -\,\tfrac{u_{1}^{4}}{24} \,-\, t_{3,1}\,u_{1} \,-\,\tfrac{t_{4,1}u_{1}^{2}}{2},
\end{equation}
and the single flat-coordinate substitution is $u_{1}\to v_{1}$. The string equation $\partial S/\partial v_{1}=0$ at $v_{1}=v_{10}$ solved for the cosmological coupling $t_{4,1}$ gives, at the physical background $t_{3,1}=0$,
\begin{equation}
t_{4,1}\;\longrightarrow\;-\frac{v_{10}^{2}}{6}.
\end{equation}

\paragraph{Amplitudes.} The structure-constant matrix
$\mathbf{C}_{1}$ is the scalar $-1$. The partition function is
\begin{equation}
Z_{0}\;=\;-\,\frac{v_{10}^{7}}{945}.
\end{equation}
Two-point diagonals (computed exactly as in FM-$x$, by twice differentiating $Z_{2}$ with respect to $t$ and substituting the string-equation solution):
\begin{equation}
Z_{12}[3,3]\;=\;-v_{10},\qquad
Z_{12}[4,4]\;=\;-\frac{v_{10}^{3}}{3}.
\end{equation}
Three-point amplitudes:
\begin{align}
Z_{123}[3,3,3] &= \tfrac{3}{v_{10}^{2}}, &
Z_{123}[3,3,4] &= \tfrac{3}{v_{10}}, \\
Z_{123}[3,4,4] &= 3, &
Z_{123}[4,4,4] &= 3\,v_{10}.
\end{align}

\paragraph{Universal terms.} The KPZ scaling at $(2,5)$ gives
$Z_{123}[k_{1},k_{2},k_{3}]\sim\mu^{(2s+3-\sum(k_{i}+2))/2}=\mu^{(7-\sum k_{i})/2}$ with $\mu\sim v_{10}^{2}$. Translating the FM-$y$ indices $(m_{i}=3,4)$ to FM-$x$ indices $(n_{i}=p-m_{i}=2,1)$, the $\mu$-power of each correlator and its universality status is:
\begin{center}
\begin{tabular}{lcccc}
\toprule
$(m_{1},m_{2},m_{3})$ & $(n_{1},n_{2},n_{3})$ & $\mu$-power of $Z_{123}$ & Non-negative integer & Universal \\
\midrule
$(4,4,4)$ & $(1,1,1)$ & $1/2$  & No  & Yes \\
$(3,4,4)$ & $(2,1,1)$ & $0$    & Yes & No \\
$(3,3,4)$ & $(2,2,1)$ & $-1/2$ & No  & Yes \\
$(3,3,3)$ & $(2,2,2)$ & $-1$   & No  & Yes \\
\bottomrule
\end{tabular}
\end{center}
The single non-universal entry, $Z_{123}[3,4,4]=3\sim\mu^{0}$, is discarded. It corresponds precisely to the fusion-forbidden $(n_{1},n_{2},n_{3})=(1,1,2)$ position of FM-$x$. The same filter applied to the one-point amplitudes removes the non-zero $Z_{1}[t_{3,1}]=v_{10}^{4}/24\sim\mu^{2}$, restoring the conformal selection rule $\langle\phi_{1,2}\rangle=0$.

\paragraph{Universal three-point ratio.} The three universal entries give
$R^{\text{FM-}y}=\{9/35,\tfrac1{35},\tfrac1{105}\}$, agreeing with the other corners and $R^{\rm CFT}$ (Table~\ref{tab:25-summary}). FM-$y$ is a Frobenius normalisation ($\varkappa=1$), and it and SC-$y$ are the same $y$-side KdV tau-function (with $Z_{0}=10u_{0}^{7}/21$), related by the resonance \eqref{eq:taudef}, which at $s=2$ is linear. The nonlinearity first appears at $s=3$ (App.~\ref{app:27-FMy}).

\subsection*{SC-$y$}
\label{app:25-SCy}

\paragraph{Spectral curve.} For Lee--Yang $(2,\pp)$ the standard
Chebyshev curve is
\begin{equation}
x(z) = 2 u_{0}\, T_{2}(z) = 4 u_{0} z^{2} - 2 u_{0}, \qquad
y(z) = 2 u_{0}^{p/2}\, T_{p}(z).
\end{equation}
At $(2,5)$:
\begin{equation}
y(z) = 2 u_{0}^{5/2}(16 z^{5} - 20 z^{3} + 5 z), \qquad
x'(z) = 8 u_{0}\, z.
\end{equation}
The map $z\mapsto x$ is two-to-one branched at $z=0$, with Galois involution $\sigma:z\mapsto -z$.

\paragraph{Topological recursion inputs.}
\begin{equation}
\omega_{0,1}(z) = y(z)\, x'(z)
= u_{0}^{7/2}\bigl(256 z^{6} - 320 z^{4} + 80 z^{2}\bigr),
\qquad
\omega_{0,2}(z_{1},z_{2}) = \frac{1}{(z_{1}-z_{2})^{2}}.
\end{equation}
Eynard--Orantin recursion at $z=0$ gives
\begin{equation}
\omega_{0,3}(z_{1},z_{2},z_{3})
= \frac{1}{80\, u_{0}^{7/2}\, z_{1}^{2} z_{2}^{2} z_{3}^{2}}.
\end{equation}

\paragraph{KdV-time amplitudes (raw $t$-basis).}
Following the standard Eynard--Orantin prescription, the matter-side $n$-point amplitudes in the matrix-model time variables $t_{k}$ are obtained as residues at $z\to\infty$:
\begin{equation}\label{eq:25-SCy-residue-formula}
Z_{k_{1}\dots k_{n}}^{(t)}
= -\frac{1}{\prod_{i}(p-2k_{i})}\,
\mathrm{Res}_{z_{1},\dots,z_{n}=\infty}
\Bigl[\omega_{0,n}(z_{1},\dots,z_{n})
\prod_{i=1}^{n} x(z_{i})^{(p-2k_{i})/2}\,dz_{i}\Bigr],
\end{equation}
where $k_{i}\in\{1,\dots,s\}$ labels the operators. All branches are taken at $z<0$ (the negative-diagonal convention also used at $(2,7)$, Appendix~\ref{app:27-explicit}). At $(2,5)$, computing these gives
\begin{align}
Z_{1}^{(t)} &= \Bigl(\tfrac{2 u_{0}^{5}}{3},\; -\tfrac{5 u_{0}^{4}}{4}\Bigr), \\[4pt]
Z_{12}^{(t)} &=
\begin{pmatrix} -\tfrac{u_{0}^3}{3} & \tfrac{u_0^2}{2} \\[2pt] \tfrac{u_0^2}{2} & -u_{0} \end{pmatrix},\\[4pt]
Z_{123}^{(t)}[1,1,1] &= -\tfrac{u_{0}}{10}, &
Z_{123}^{(t)}[1,1,2] &= \tfrac{1}{10}, \\
Z_{123}^{(t)}[1,2,2] &= -\tfrac{1}{10\, u_{0}}, &
Z_{123}^{(t)}[2,2,2] &= \tfrac{1}{10\, u_{0}^{2}}.
\end{align}
The partition function is obtained by integrating the first derivative  $Z_1^{(t)}[1]=2u_0^5/3$ with respect to $t_1$ and using $t_{1}^{*}(u_{0}) = -5 u_{0}^{2}/2$, which follows from~\eqref{kdvdef}
\begin{equation}
Z_{0} = \int_{0}^{u_{0}} \frac{\partial t_{1}^{*}}{\partial u_{0}'}\, Z_{1}^{(t)}[1]\big|_{u_{0}\to u_{0}'}\, du_{0}'
= \frac{10\, u_{0}^{7}}{21}.
\end{equation}

\paragraph{Chain rule to physical $\tau$-couplings.}
At $(2,5)$ the BZ/Artemev resonance map \eqref{eq:taudef} is purely linear (no $\tau$-mixing, since $s=2$):
\begin{equation}\label{eq:25-tau-t-relation}
t_{1} = (2s+1)\,u_{0}^{2}\,\tau_{1},
\qquad
t_{2} = (2s+1)\,u_{0}^{3}\,\tau_{2},
\end{equation}
i.e.\ $t_{k} = 5\, u_{0}^{k+1}\,\tau_{k}$. The chain rule for the $n$-th derivative is therefore multiplicative,
\begin{equation}
\frac{\partial^{n}F_{0}}
     {\partial\tau_{k_{1}}\cdots\partial\tau_{k_{n}}}
= \prod_{i=1}^{n}\bigl[(2s+1)\,u_{0}^{k_{i}+1}\bigr]\cdot
  Z_{k_{1}\dots k_{n}}^{(t)}
= (2s+1)^{n}\, u_{0}^{\sum_{i}(k_{i}+1)}\, Z_{k_{1}\dots k_{n}}^{(t)}.
\end{equation}
Each operator insertion additionally carries the per-insertion factor $\varkappa=-2$ of \eqref{eq:kappa-value}, which converts the spectral-curve normalisation to the Frobenius one. Artemev's amplitude (his eq.~\eqref{eq:agndef}), dressed with this factor, is then
\begin{equation}
A^{0}_{n}(k_{1},\dots,k_{n})
= \varkappa^{n}\,u_{0}^{-\sum_{i}(k_{i}+1)}\cdot
\frac{\partial^{n}F_{0}}
     {\partial\tau_{k_{1}}\cdots\partial\tau_{k_{n}}}
= \bigl(\varkappa\,(2s+1)\bigr)^{n}\, Z_{k_{1}\dots k_{n}}^{(t)},
\end{equation}
while the partition function $A^{0}_{0}=\varkappa\,Z_{0}$ inherits a single factor of $\varkappa$ through the dilaton relation $Z_{0}=\int(\partial_{u_{0}}t_{1}^{*})\,Z_{1}^{(t)}[1]\,du_{0}$. At $(2,5)$, $\varkappa(2s+1) = -10$, so the chain factors are $-10,\,100,\,-1000$:
\begin{align}
A^{0}_{3}[1,1,1] &= (-1000)\bigl(-\tfrac{u_{0}}{10}\bigr) = 100\,u_{0}, &
A^{0}_{3}[1,1,2] &= (-1000)\bigl(\tfrac{1}{10}\bigr) = -100, \\
A^{0}_{3}[1,2,2] &= \tfrac{100}{u_{0}}, &
A^{0}_{3}[2,2,2] &= -\tfrac{100}{u_{0}^{2}}, \\
A^{0}_{2}[1,1] &= 100\bigl(-\tfrac{u_{0}^{3}}{3}\bigr) = -\tfrac{100\, u_{0}^{3}}{3}, &
A^{0}_{2}[2,2] &= -100\, u_{0}.
\end{align}

\paragraph{Singular projection (universality filter).}
The singular projection of an amplitude is, by convention \cite{Belavin:2008kv,Artemev}, the part proportional to negative or odd-positive integer powers of $u_{0}$. Non-negative even integer powers are non-universal contributions to the partition function and are discarded:
\begin{equation}\label{eq:25-sing-rule}
A^{0}_{n}\bigl|_{\text{sing}}
= \begin{cases} A^{0}_{n} & \text{if power of $u_{0}$ is $<0$ or odd $>0$,}\\
                0 & \text{otherwise (power $=0$ or even $>0$).} \end{cases}
\end{equation}
Applied to the four inequivalent entries of $A^{0}_{3}$ at $(2,5)$:
\begin{center}
\begin{tabular}{lcccl}
\toprule
$(k_{1},k_{2},k_{3})$ & $A^{0}_{3}$ & $u_{0}$-power & Universal? & After sing.\ projection \\
\midrule
$(1,1,1)$ & $\;100\, u_{0}$        & $+1$ & yes (odd-positive) & $100\, u_{0}$ \\
$(1,1,2)$ & $-100$                  & $0$  & no (zero)    & $0$ \\
$(1,2,2)$ & $\;\tfrac{100}{u_{0}}$         & $-1$ & yes (negative)      & $\tfrac{100}{u_{0}}$ \\
$(2,2,2)$ & $-\tfrac{100}{u_{0}^{2}}$      & $-2$ & yes (negative)      & $-\tfrac{100}{u_{0}^{2}}$ \\
\bottomrule
\end{tabular}
\end{center}
The single non-universal entry is $(1,1,2)$, which corresponds precisely (under the Kac dictionary $\phi_{1,k}\equiv\phi_{1,p-k}$) to the fusion-forbidden FM-$x$ triple $(1,1,2)$. At $(2,5)$ it sits at the projected-out power $u_{0}^{0}$, so it is removed by the universality filter alone, independently of $\varkappa$ (at $(2,7)$, Appendix~\ref{app:27-explicit}, the analogous forbidden entry instead sits at a surviving odd power and needs $\varkappa=-2$).

\paragraph{Fusion matrices from the normalised amplitudes.}
All three-point amplitudes carry a common overall factor: the square of the two-point normalisation, $N=\bigl(\varkappa(2s+1)\bigr)^{2}=100$ (indeed $A^{0}_{3}(1,1,1)|_{u_{0}=1}=N$). The reduced amplitudes $A^{0}_{3}|_{\text{sing}}/N$ are therefore integers at $u_{0}=1$; grouping them by the third index $k_{3}$ gives, for each primary, the matrix $\bigl(\mathcal N_{k_{3}}\bigr)_{k_{1}k_{2}} =A^{0}_{3}(k_{1},k_{2},k_{3})|_{\text{sing}}/N$,
\begin{equation}\label{eq:25-A03-verlinde}
\frac{A^{0}_{3}|_{\text{sing}}}{\varkappa^{2}(2s+1)^{2}}\bigg|_{u_{0}=1}
= \underbrace{\begin{pmatrix}1 & 0\\ 0 & 1\end{pmatrix}}_{\mathcal N_{1}}\,\delta_{k_{3},1}
\;+\;
\underbrace{\begin{pmatrix}0 & 1\\ 1 & -1\end{pmatrix}}_{\mathcal N_{2}}\,\delta_{k_{3},2}.
\end{equation}
These are precisely the fusion matrices of the two $(2,5)$ primaries: $\mathcal N_{1}=\mathbb{1}$ (the identity $\phi_{1,1}$ fuses trivially), and $\mathcal N_{2}$, the fusion matrix of $\phi_{1,2}$, whose entry $(\mathcal N_{2})_{22}=-1$ is the signed Verlinde number characteristic of the non-unitary Lee--Yang model. The spectral-curve amplitudes thus reproduce the signed Verlinde algebra, sign included.

\paragraph{Universal three-point ratio.}
The universal ratio combines the three-, two- and zero-point amplitudes into a quantity that is invariant under any rescaling of the operators $O_{k}\to c_{k}O_{k}$ (hence scheme-independent) and carries zero $u_{0}$-dimension (hence a pure number):
\begin{equation}\label{eq:25-R-SCy-formula}
R^{\text{SC-}y}(k_{1},k_{2},k_{3})
:= \frac{(A^{0}_{3}|_{\text{sing}})^{2}\,A^{0}_{0}}
{A^{0}_{2}[k_{1},k_{1}]\,A^{0}_{2}[k_{2},k_{2}]\,A^{0}_{2}[k_{3},k_{3}]}\,.
\end{equation}
It is not, however, blind to the per-insertion factor $\varkappa$. An $n$-point amplitude carries $\varkappa^{n}$ (with $A^{0}_{0}=\varkappa Z_{0}$), so the ratio scales as $(\varkappa^{3})^{2}\,\varkappa/(\varkappa^{2})^{3}=\varkappa$: a single net power of $\varkappa$ survives. Thus $R$ discards the operator normalisation but retains the one physical per-insertion factor, and this is what pins it down: at the naive $\varkappa=1$ the two spectral-curve corners would give $\pm\tfrac12 R^{\text{CFT}}$, whereas $\varkappa=-2$ makes $R^{\text{SC-}y}$ equal $R^{\text{CFT}}$ exactly. The $u_{0}$-powers cancel triple by triple, and the four inequivalent entries are
\begin{equation}\label{eq:25-R-SCy}
R^{\text{SC-}y}=\Bigl\{\tfrac{9}{35},\,0,\,\tfrac{1}{35},\,\tfrac{1}{105}\Bigr\}
\quad\text{for}\quad
(k_{1},k_{2},k_{3})=(1,1,1),(1,1,2),(1,2,2),(2,2,2),
\end{equation}
in exact agreement with FM-$x$, FM-$y$, SC-$x$ and the CFT prediction.

\subsection*{SC-$x$}
\label{app:25-SCx}

\paragraph{Spectral curve.} For Lee--Yang $(2,\pp)$ the swapped
Chebyshev curve is
\begin{equation}\label{eq:25-SCx-curve}
\check x(z) = 2\,u_{0}^{p/2}\,T_{p}(z), \qquad
\check y(z) = 2\,u_{0}\,T_{2}(z) = 4\,u_{0}\,z^{2} - 2\,u_{0}.
\end{equation}
At $(2,5)$, $p=5$:
\begin{equation}
\check x(z) = 2\,u_{0}^{5/2}\bigl(16 z^{5} - 20 z^{3} + 5 z\bigr), \qquad
\check y(z) = 4\,u_{0}\,z^{2} - 2\,u_{0}.
\end{equation}
The map $z\mapsto\check x$ has degree $p=5$, branched at the $p-1 = 4$ interior ramification points $\zeta_{m}=\cos(\pi m/p)$, $m=1,\dots,4$:
\begin{equation}
\zeta_{1}=\tfrac{1+\sqrt{5}}{4},\quad
\zeta_{2}=\tfrac{-1+\sqrt{5}}{4},\quad
\zeta_{3}=-\zeta_{2},\quad
\zeta_{4}=-\zeta_{1}.
\end{equation}
Near each $\zeta_{m}$ the local Galois involution $\bar{z}^{(m)}(z)$ satisfies $\bar z^{(m)}(\zeta_{m})=\zeta_{m}$ and $d\bar z^{(m)}/dz\big|_{z=\zeta_{m}}=-1$; the higher-order behaviour of $\bar z^{(m)}$ is not needed (see below).

\paragraph{Topological recursion inputs.}
\begin{equation}
\check\omega_{0,1}(z) = \check y(z)\,\check x'(z),
\qquad
\check\omega_{0,2}(z_{1},z_{2}) = \frac{1}{(z_{1}-z_{2})^{2}}.
\end{equation}
The key algebraic simplification at the swap is
\begin{equation}\label{eq:25-SCx-ykey}
\check y(z) - \check y(\bar z^{(m)})
= 4\,u_{0}\bigl(z-\bar z^{(m)}\bigr)\bigl(z+\bar z^{(m)}\bigr),
\end{equation}
in which the factor $(z-\bar z^{(m)})$ cancels with the corresponding factor from $\int_{\bar z}^{z}\check\omega_{0,2}(z_{0},\cdot)$ in the recursion kernel. The Eynard--Orantin recursion, with residues collected at the four $\zeta_{m}$, gives the closed form
\begin{equation}\label{eq:25-SCx-omega03}
\check\omega_{0,3}(z_{1},z_{2},z_{3})
= -\frac{1}{8\,u_{0}^{7/2}}\sum_{m=1}^{4}
\frac{c_{m}}{(z_{1}-\zeta_{m})^{2}(z_{2}-\zeta_{m})^{2}(z_{3}-\zeta_{m})^{2}},\quad
c_{m} = \frac{(-1)^{m}}{50}\,\frac{\sin^{2}(\pi m/5)}{\cos(\pi m/5)}.
\end{equation}
Explicitly, $c_{1}=c_{4}=\tfrac{1}{40}-\tfrac{3\sqrt{5}}{200}$ and $c_{2}=c_{3}=\tfrac{1}{40}+\tfrac{3\sqrt{5}}{200}$.

\paragraph{KP-time amplitudes ($\check t$-basis).}
The SC-$x$ amplitudes in the KP-time variables $\check t_{k}$ are obtained from Artemev's residue formula~\eqref{eq:SCx-kpderiv} with fractional powers of $\check x$:
\begin{equation}\label{eq:25-SCx-residue-formula}
\check Z_{k_{1}\dots k_{n}}^{(\check t)}
= -\frac{1}{\prod_{i}(p-2k_{i})}\,
\mathrm{Res}_{z_{1},\dots,z_{n}=\infty}
\Bigl[\check\omega_{0,n}(z_{1},\dots,z_{n})
\prod_{i=1}^{n}\check x(z_{i})^{(p-2k_{i})/p}\,dz_{i}\Bigr],
\end{equation}
where $k_{i}\in\{1,\dots,s\}$ labels the operators. At $(2,5)$ this yields
\begin{align}
\check Z_{1}^{(\check t)} &= \Bigl(\tfrac{2\,u_{0}^{5}}{3},\;0\Bigr), \\[4pt]
\check Z_{12}^{(\check t)} &=
\begin{pmatrix} -\tfrac{u_{0}^{3}}{3} & 0 \\[2pt] 0 & -u_{0} \end{pmatrix},\\[4pt]
\check Z_{123}^{(\check t)}[1,1,1] &= -\tfrac{u_{0}}{10}, &
\check Z_{123}^{(\check t)}[1,1,2] &= 0, \\
\check Z_{123}^{(\check t)}[1,2,2] &= -\tfrac{1}{10\,u_{0}}, &
\check Z_{123}^{(\check t)}[2,2,2] &= \tfrac{1}{10\,u_{0}^{2}}.
\end{align}
Three structural features distinguish the swapped side from SC-$y$:
\begin{enumerate}
\item Only $\check t_{1}$ has a nonzero background value. From
\begin{equation}
\check t_{k}^{0} = \tfrac{1}{2}\,\mathrm{Res}_{z=\infty}\bigl[
\check\omega_{0,1}(z)\,\check x(z)^{(-p+2k)/p}\bigr],
\end{equation}
one finds $\check t_{1}^{0}(u_{0}) = 5\,u_{0}^{2}/2$ and $\check t_{2}^{0}=0$. This is the ``no resonance'' feature of the swap: the dual Liouville coupling and the dual KP-time agree, modulo the single $\check t_{1}$ background value, with no $\tau$-mixing required.
\item The two-point matrix $\check Z_{12}^{(\check t)}$ is diagonal:
no off-diagonal mixing to be removed.
\item The fusion-forbidden three-point entry $(1,1,2)$ vanishes
automatically, with no singular projection needed.
\end{enumerate}

The partition function is obtained by integrating $\check Z_{1}^{(\check t)}[1] = 2 u_{0}^{5}/3$ along $u_{0}$ via the chain rule $\partial\check{\mathcal F}_{0}/\partial u_{0} = (\partial\check t_{1}^{0}/\partial u_{0})\,\check Z_{1}^{(\check t)}[1]$:
\begin{equation}\label{eq:25-SCx-Z0}
\check Z_{0} = \int_{0}^{u_{0}}
\frac{\partial \check t_{1}^{0}}{\partial u_{0}'}\,
\check Z_{1}^{(\check t)}[1]\big|_{u_{0}\to u_{0}'}\,du_{0}'
= \int_{0}^{u_{0}} 5\,u_{0}'\cdot\tfrac{2\,u_{0}'^{5}}{3}\,du_{0}'
= \frac{10\,u_{0}^{7}}{21},
\end{equation}
which agrees with $Z_{0}$ on the SC-$y$ side, as required by the underlying free-energy identification.

\paragraph{Universal three-point ratio.}
The swapped side needs no resonance and no singular projection, but it carries the same per-insertion factor $\varkappa=-2$ \eqref{eq:kappa-value} as SC-$y$, since both are spectral-curve (rather than Frobenius) normalisations. With $\check A^{0}_{n}= \varkappa^{n}\check Z^{(\check t)}_{k_{1}\dots k_{n}}$ and $\check A^{0}_{0}=\varkappa\check Z_{0}$, the universal ratio is
\begin{equation}\label{eq:25-R-SCx-formula}
R^{\text{SC-}x}(k_{1},k_{2},k_{3})
:= \frac{\bigl(\check A^{0}_{3}\bigr)^{2}\,\check A^{0}_{0}}
{\check A^{0}_{2}[k_{1},k_{1}]\,
\check A^{0}_{2}[k_{2},k_{2}]\,
\check A^{0}_{2}[k_{3},k_{3}]}
= \varkappa\,
\frac{\bigl(\check Z_{123}^{(\check t)}\bigr)^{2}\,\check Z_{0}}
{\check Z_{12}^{(\check t)}[k_{1},k_{1}]\,
\check Z_{12}^{(\check t)}[k_{2},k_{2}]\,
\check Z_{12}^{(\check t)}[k_{3},k_{3}]}\,,
\end{equation}
again with a single net $\varkappa$ from $\check A^{0}_{0}$. The $u_{0}$-powers cancel triple-by-triple, and the four inequivalent entries are
\begin{equation}\label{eq:25-R-SCx}
R^{\text{SC-}x}=\Bigl\{\tfrac{9}{35},\,0,\,\tfrac{1}{35},\,\tfrac{1}{105}\Bigr\}
\quad\text{for}\quad
(k_{1},k_{2},k_{3})=(1,1,1),(1,1,2),(1,2,2),(2,2,2),
\end{equation}
with the fusion-forbidden $(1,1,2)$ entry vanishing automatically (the two-point matrix is already diagonal and $\check Z_{123}^{(\check t)}[1,1,2]=0$ on the swapped curve).

With the common factor $\varkappa=-2$,
\begin{equation}\label{eq:25-R-SCx-vs-SCy}
R^{\text{SC-}x} = R^{\text{SC-}y} = R^{\text{CFT}}
\qquad
\text{(on every fusion-allowed triple).}
\end{equation}

\subsection*{Summary}
\label{app:25-summary}

\paragraph{Comparison of the four formulations.} Table~\ref{tab:25-summary}
collects the universal three-point ratio $R(k_{1},k_{2},k_{3})$ at $(2,5)$ in the four corners of the diagram. With the per-insertion factor $\varkappa=-2$ on the two spectral-curve sides \eqref{eq:kappa-value}, all four corners agree exactly with the CFT prediction, on every fusion-allowed triple.

\begin{table}[h]
\centering
\renewcommand{\arraystretch}{1.3}
\begin{tabular}{lcccc}
\toprule
$(n_{1},n_{2},n_{3})$ & FM-$x$ & FM-$y$ & SC-$x$ & SC-$y$ \\
 & ($\varkappa=1$) & ($\varkappa=1$) & ($\varkappa=-2$) & ($\varkappa=-2$) \\
\midrule
$(1,1,1)$ & $9/35$  & $9/35$               & $9/35$               & $9/35$  \\
$(1,1,2)$ & $0$     & $0$ (sing.\ proj.)   & $0$ (automatic)      & $0$ (sing.\ proj.) \\
$(1,2,2)$ & $1/35$  & $1/35$               & $1/35$               & $1/35$  \\
$(2,2,2)$ & $1/105$ & $1/105$              & $1/105$              & $1/105$ \\
\midrule
CFT prediction & \multicolumn{4}{c}{$\bigl\{\,9/35,\;0,\;1/35,\;1/105\,\bigr\}$} \\
\bottomrule
\end{tabular}
\caption{Universal three-point ratio $R$ at $(2,5)$ in the four corners
of the diagram. The Frobenius sides (FM-$x$, FM-$y$) use $\varkappa=1$;
the spectral-curve sides (SC-$x$, SC-$y$) carry the per-insertion
conversion factor $\varkappa=-2$. All four agree with the CFT values.}
\label{tab:25-summary}
\end{table}

At $\varkappa=1$ the two SC corners would read $\pm\tfrac12 R^{\text{CFT}}$; the single net power of $\varkappa$ in the ratio (from $A^{0}_{0}=\varkappa Z_{0}$) fixes both the $\tfrac12$ and the sign.

\paragraph{Generating functions.}
Collecting the partition function, one-point amplitudes, two-point matrix and three-point amplitudes from each of the four subsections, the genus-zero free energy at $(2,5)$ reads
\begin{equation}\label{eq:25-F0-master}
F_{0}(t)\;=\;Z_{0}\;+\;\sum_{a}Z_{1}[a]\,\delta t_{a}
\;+\;\tfrac{1}{2}\sum_{a,b}Z_{12}[a,b]\,\delta t_{a}\delta t_{b}
\;+\;\tfrac{1}{6}\sum_{a,b,c}Z_{123}[a,b,c]\,\delta t_{a}\delta t_{b}\delta t_{c}
\;+\;\mathcal{O}(\delta t^{4}),
\end{equation}
where $\delta t_{a}:=t_{a}-t_{a}^{*}$ is the deviation from the MLG background.

\noindent The two FM generating functions are obtained the same way; at the MLG background $Z_{0}^{\text{FM-}x}=-v_{10}^{7}/65625$ and $Z_{0}^{\text{FM-}y}=-v_{10}^{7}/945$, and their filtered forms are related by the dilatation \eqref{eq:25-FM-relation}. The two spectral-curve generating functions are

\noindent\textbf{SC-$y$} (variables $t_{1}, t_{2}$; background $t_{1}^{*}=-5\,u_{0}^{2}/2$, $t_{2}^{*}=0$; the universality filter would discard $Z_{1}^{(t)}[2]\,t_{2}$, $Z_{12}^{(t)}[1,2]\,(t_{1}+5 u_{0}^{2}/2)t_{2}$ and $Z_{123}^{(t)}[1,1,2]\,(t_{1}+5 u_{0}^{2}/2)^{2}t_{2}$, marked $(\star)$):
\begin{equation}\label{eq:25-F0-SCy}
\begin{aligned}
F_{0}^{\text{SC-}y} ={}& \frac{10\,u_{0}^{7}}{21}
+\tfrac{2 u_{0}^{5}}{3}\bigl(t_{1}+\tfrac{5 u_{0}^{2}}{2}\bigr)
- \tfrac{5 u_{0}^{4}}{4}\,t_{2}\;\;(\star) \\
&{} - \tfrac{1}{2}\,\tfrac{u_{0}^{3}}{3}\bigl(t_{1}+\tfrac{5 u_{0}^{2}}{2}\bigr)^{2}
+ \tfrac{u_{0}^{2}}{2}\bigl(t_{1}+\tfrac{5 u_{0}^{2}}{2}\bigr)t_{2}\;\;(\star)
- \tfrac{1}{2}\,u_{0}\,t_{2}^{\,2} \\
&{} + \tfrac{1}{6}\!\Bigl[-\tfrac{u_{0}}{10}\bigl(t_{1}+\tfrac{5 u_{0}^{2}}{2}\bigr)^{3}
+\tfrac{1}{10\,u_{0}^{2}}\,t_{2}^{\,3}\Bigr] \\
&{} + \tfrac{1}{2}\cdot\tfrac{1}{10}\bigl(t_{1}+\tfrac{5 u_{0}^{2}}{2}\bigr)^{2}t_{2}\;\;(\star)
- \tfrac{1}{2}\,\tfrac{1}{10\,u_{0}}\bigl(t_{1}+\tfrac{5 u_{0}^{2}}{2}\bigr)t_{2}^{\,2}
+ \mathcal{O}(t^{4}).
\end{aligned}
\end{equation}

\noindent\textbf{SC-$x$} (variables $\check t_{1}, \check t_{2}$; background $\check t_{1}^{*}=5\,u_{0}^{2}/2$, $\check t_{2}^{*}=0$; no $(\star)$ terms appear: all entries marked $(\star)$ on the other sides vanish automatically here):
\begin{equation}\label{eq:25-F0-SCx}
\begin{aligned}
F_{0}^{\text{SC-}x} ={}& \frac{10\,u_{0}^{7}}{21}
+ \frac{2\,u_{0}^{5}}{3}\bigl(\check t_{1}-\tfrac{5 u_{0}^{2}}{2}\bigr)
- \tfrac{1}{2}\,\tfrac{u_{0}^{3}}{3}\bigl(\check t_{1}-\tfrac{5 u_{0}^{2}}{2}\bigr)^{2}
- \tfrac{1}{2}\,u_{0}\,\check t_{2}^{\,2} \\
&{} + \tfrac{1}{6}\!\Bigl[-\tfrac{u_{0}}{10}\bigl(\check t_{1}-\tfrac{5 u_{0}^{2}}{2}\bigr)^{3}
+\tfrac{1}{10\,u_{0}^{2}}\,\check t_{2}^{\,3}\Bigr]
- \tfrac{1}{2}\,\tfrac{1}{10\,u_{0}}\bigl(\check t_{1}-\tfrac{5 u_{0}^{2}}{2}\bigr)\,\check t_{2}^{\,2}
+ \mathcal{O}(\check t^{4}).
\end{aligned}
\end{equation}

The expressions \eqref{eq:25-F0-SCy}--\eqref{eq:25-F0-SCx} are the unfiltered spectral-curve generating functions.

The conjecture is that $F_{0}^{\text{SC/FM-}y}$ and $F_{0}^{\text{SC/FM-}x}$ are related by a Kontsevich-type kernel transformation $\mathcal{K}$ acting on the generating function itself, $F_{0}^{\text{SC/FM-}x} = \mathcal{K}\bigl[ F_{0}^{\text{SC/FM-}y}\bigr]$, modulo $(\star)$ terms.

\paragraph{Relations between the four generating functions.}
On the FM side, we obtain after the universality filter applied to $F_{0}^{\text{FM-}y}$:
\begin{equation}\label{eq:25-FM-relation}
\;F_{0}^{\text{FM-}y,\text{filt}}(t_{4,1},t_{3,1})
\;=\; \frac{625}{9}\,
F_{0}^{\text{FM-}x}\!\Bigl(-\tfrac{3}{5}\,t_{4,1},\;-\tfrac{3}{25}\,t_{3,1}\Bigr).
\end{equation}

On the SC side, in the common deviation-from-background variables $\delta t_{a} := t_{a} - t_{a}^{*}$ (with $t_{1}^{*}=-5u_{0}^{2}/2$ for SC-$y$ and $\check t_{1}^{*}=+5u_{0}^{2}/2$ for SC-$x$), the two pre-filter free energies differ only by the $(\star)$-marked non-universal terms of $F_{0}^{\text{SC-}y}$:
\begin{equation}\label{eq:25-SC-sumZ0}
F_{0}^{\text{SC-}y}(\delta t_{1},\delta t_{2})
- F_{0}^{\text{SC-}x}(\delta t_{1},\delta t_{2})
= -\tfrac{5 u_{0}^{4}}{4}\delta t_{2}
+ \tfrac{u_{0}^{2}}{2}\delta t_{1}\delta t_{2}
+ \tfrac{1}{20}\delta t_{1}^{\,2}\delta t_{2}\,,
\end{equation}
i.e.\ their universal (fusion-allowed) parts coincide identically. This is the $(2,5)$ instance of the genus-zero SC-$y\leftrightarrow$SC-$x$ kernel \eqref{eq:27-SCkernel}: the rescaling of the flow times is trivial here ($\Lambda=1$, since $Z_{12}^{(t)}[1,1]=\check Z_{12}^{(\check t)}[1,1]=-u_{0}^{3}/3$), whereas at $(2,7)$ it is the dilatation $\Lambda=\mathrm{diag}(\tfrac32,1,1)$.

\section{Explicit calculation for $(2,7)$}
\label{app:27-explicit}

In this appendix we consider the $(q,p)=(2,7)$ model ($s=3$, three matter primaries $\phi_{1,1},\phi_{1,2},\phi_{1,3}$). The cosmological background is $v_{1}=v_{10}$ on the FM sides and $u=u_{0}$ on the SC sides. The universality filter discards correlators proportional to non-negative even integer powers of $u_{0}$. The new feature relative to $(2,5)$ is that the resonance is genuinely nonlinear and the per-insertion factor $\varkappa=-2$ is required for the cancellation: it cancels a fusion-forbidden three-point amplitude that no longer sits at a projected-out power.

\subsection*{FM-$x$}\label{app:27-FMx}

In FM-$x$ at $(2,7)$ the relevant Frobenius manifold is $A_{p-1}=A_{6}$, with polynomial
\begin{equation}
P(x) \;=\; x^{7} + u_{1}x^{5} + u_{2}x^{4} + u_{3}x^{3} + u_{4}x^{2}
+ u_{5}x + u_{6}.
\end{equation}
Active couplings are $t_{1,1},t_{1,2},t_{1,3}$, the three physical Liouville couplings. With $c(x)=\Gamma(x-\lfloor x\rfloor)/\Gamma(x+1)$, the action, which is the coefficient of $x^{-1}$ in the Laurent expansion of $-c(\tfrac{\pp+2}{\pp})P^{(\pp+2)/\pp} -\sum_{n=1}^{s}c(\tfrac{\pp-2n}{\pp})P^{(\pp-2n)/\pp}\,t_{1,n}$, evaluates after Gamma-function identities to a rational polynomial,
\begin{equation}
\begin{aligned}
S ={}& \tfrac{19}{686}u_{1}^{5} - \tfrac{15}{49}u_{1}^{2}u_{2}^{2}
- \tfrac{10}{49}u_{1}^{3}u_{3} + \tfrac{5}{14}u_{2}^{2}u_{3}
+ \tfrac{5}{14}u_{1}u_{3}^{2} + \tfrac{5}{7}u_{1}u_{2}u_{4}
- \tfrac12 u_{4}^{2} + \tfrac{5}{14}u_{1}^{2}u_{5}
- u_{3}u_{5} - u_{2}u_{6}\\
&{} + \bigl(-\tfrac{3}{49}u_{1}^{3}+\tfrac17 u_{2}^{2}+\tfrac27 u_{1}u_{3}-u_{5}\bigr)t_{1,1}
+ \bigl(\tfrac27 u_{1}^{2}-u_{3}\bigr)t_{1,2}
- u_{1}\,t_{1,3}.
\end{aligned}
\end{equation}

\paragraph{Flat coordinates.}
Requiring the coefficient of $x^{-k}$ in $P^{1/\pp}-\tfrac1\pp\sum_{n}v_{n}P^{-n/\pp}-x$ to vanish for $k=1,\dots,6$ gives
\begin{align}
&u_{1}\!\to v_{1},\;\;
u_{2}\!\to v_{2},\;\;
u_{3}\!\to \tfrac{2v_{1}^{2}}{7}+v_{3},\;\;
u_{4}\!\to \tfrac{3v_{1}v_{2}}{7}+v_{4},\;\;\\
&u_{5}\!\to \tfrac{v_{1}^{3}+14 v_{1}v_{3}+7v_{2}^{2}}{49}+v_{5},\;\;
u_{6}\!\to \tfrac{v_{1}^{2}v_{2}+7v_{2}v_{3}+7v_{1}v_{4}}{49}+v_{6}.
\end{align}
At the MLG slice $v_{2}=\dots=v_{6}=0$,
\begin{equation}\label{eq:27-Q-MLG}
P(x)\big|_{\rm MLG}=x^{7}+v_{10}x^{5}+\tfrac{2v_{10}^{2}}{7}x^{3}
+\tfrac{v_{10}^{3}}{49}x,
\qquad
P(x)\big|_{\rm MLG,\,v_{10}=-7}=2\,T_{7}(x/2),
\end{equation}
the explicit FM-$x$/SC-$x$ bridge of footnote~\ref{fn:cheb-norm}.

\paragraph{String equation and structure constants.}
Solving $\partial S/\partial v_{k}=0$ ($k=2,\dots,6$) at $v_{k\ge2}=0$ gives
\begin{equation}
t_{1,1}\to\tfrac{v_{10}^{2}}{14},\qquad t_{1,2}\to0,\qquad t_{1,3}\to0.
\end{equation}
The Saito pairing is the antidiagonal $\eta_{ij}=-\delta_{i+j,7}$, and the top structure-constant matrix at the MLG slice is diagonal,
\begin{equation}
\mathbf{C}_{6}\big|_{\rm MLG}
=\mathrm{diag}\!\Bigl(\tfrac{v_{10}^{5}}{16807},\,-\tfrac{v_{10}^{4}}{2401},\,
\tfrac{v_{10}^{3}}{343},\,-\tfrac{v_{10}^{2}}{49},\,\tfrac{v_{10}}{7},\,-1\Bigr)
=\mathrm{diag}\bigl((-v_{10}/7)^{6-\ell}\bigr)_{\ell=1}^{6}.
\end{equation}

\paragraph{Partition function and amplitudes.}
Proceeding exactly as at $(2,5)$ (App.~\ref{app:25-FMx}), the partition function, one- and two-point amplitudes are
\begin{equation}
Z_{0}=-\frac{v_{10}^{9}}{37059435},\qquad
Z_{1}=\Bigl(-\tfrac{v_{10}^{7}}{588245},0,0\Bigr),\qquad
Z_{12}=\mathrm{diag}\!\Bigl(-\tfrac{v_{10}^{5}}{12005},\,
-\tfrac{v_{10}^{3}}{147},\,-v_{10}\Bigr),
\end{equation}
(off-diagonal two-point entries vanish by parity), and the fusion-allowed three-point amplitudes are
\begin{align}
Z_{123}[1,1,1]&=-\tfrac{v_{10}^{3}}{343}, &
Z_{123}[1,2,2]&=-\tfrac{v_{10}}{7}, &
Z_{123}[1,3,3]&=-\tfrac{7}{v_{10}}, \nonumber\\
Z_{123}[2,2,3]&=-\tfrac{7}{v_{10}}, &
Z_{123}[2,3,3]&=-\tfrac{49}{v_{10}^{2}}, &
Z_{123}[3,3,3]&=-\tfrac{343}{v_{10}^{3}},
\end{align}
all fusion-forbidden entries ($[1,1,2],[1,1,3],[1,2,3],[2,2,2]$) vanishing automatically in FM-$x$.

\paragraph{Universal three-point ratio.}
With $R[n_{1},n_{2},n_{3}]=(Z_{123})^{2}Z_{0}/(Z_{12}Z_{12}Z_{12})$ the six inequivalent fusion-allowed values are
\begin{equation}
R^{\rm FM\text{-}x}=\Bigl\{\tfrac{25}{63},\,\tfrac17,\,\tfrac{1}{63},\,
\tfrac{1}{35},\,\tfrac{1}{105},\,\tfrac{1}{315}\Bigr\}
\;\text{for}\;
(1,1,1),(1,2,2),(1,3,3),(2,2,3),(2,3,3),(3,3,3),
\end{equation}
in agreement with the CFT prediction $R^{\rm CFT}=(p-2n_{1})(p-2n_{2})(p-2n_{3})/\bigl((p-2)p(p+2)\bigr)\, \theta(n_{1},n_{2},n_{3})$.

\subsection*{FM-$y$}\label{app:27-FMy}

In FM-$y$ at $(2,7)$ the Frobenius manifold is $A_{q-1}=A_{1}$, with $Q(y)=y^{2}+u_{1}$. The active couplings are $t_{m,1}$ with $pn<qm$, i.e.\ $m\in\{4,5,6\}$, mapped to FM-$x$ by the Kac dictionary $t_{m,1}^{\rm FM\text{-}y}\!\leftrightarrow\!t_{1,p-m}^{\rm FM\text{-}x}$: the cosmological $t_{6,1}\!\leftrightarrow\!t_{1,1}$, and $t_{5,1}\!\leftrightarrow\!t_{1,2}$, $t_{4,1}\!\leftrightarrow\!t_{1,3}$. The action is
\begin{equation}
S=-\tfrac{1}{120}u_{1}^{5}-u_{1}t_{4,1}-\tfrac12 u_{1}^{2}t_{5,1}
-\tfrac16 u_{1}^{3}t_{6,1},
\end{equation}
with flat-coordinate substitution $u_{1}\to v_{1}$ and string-equation solution for the cosmological coupling $t_{6,1}\to-(v_{10}^{4}+24t_{4,1}+24v_{10}t_{5,1})/(12v_{10}^{2})$, i.e.\ $t_{6,1}\to-v_{10}^{2}/12$ at the physical background $t_{4,1}=t_{5,1}=0$. The $A_{1}$ structure constant is the scalar $-1$ and the partition function is
\begin{equation}
Z_{0}=-\frac{v_{10}^{9}}{45360}.
\end{equation}

\paragraph{Resonance and physical correlators.}
As at $(2,5)$, FM-$y$ and SC-$y$ are two descriptions of the same $y$-side KdV tau-function: the $A_{1}$ free energy reproduces the $y$-side $t$-basis amplitudes of Appendix~\ref{app:27-explicit}, and the physical multipoint correlators are obtained from it by the resonance map $\tau\!\to\!t(\tau)$ of \eqref{eq:taudef}. At $s=3$ this map is nonlinear: it contains the single quadratic term $\tfrac12\tau_{1}^{2}$ in $t_{3}$, eq.~\eqref{eq:res-27}, and the per-insertion factor $\varkappa=-2$ is required to cancel the surviving fusion-forbidden $(1,1,3)$ amplitude, eq.~\eqref{eq:A03-27}. The resulting amplitudes, signed Verlinde matrices \eqref{eq:verl-27-fixed} and universal ratios
\begin{equation}
R^{\rm FM\text{-}y}=\Bigl\{\tfrac{25}{63},\,\tfrac17,\,\tfrac{1}{63},\,
\tfrac{1}{35},\,\tfrac{1}{105},\,\tfrac{1}{315}\Bigr\}
\;\text{for}\;
(1,1,1),(1,2,2),(1,3,3),(2,2,3),(2,3,3),(3,3,3),
\end{equation}
therefore coincide with SC-$y$ and FM-$x$. (Being a Frobenius normalisation, FM-$y$ itself carries $\varkappa=1$, and the factor $\varkappa=-2$ enters only through the identification of its $y$-side data with the spectral-curve amplitudes of Appendix~\ref{app:27-explicit}, where it performs the $(1,1,3)$ cancellation.)

\subsection*{SC-$y$}\label{app:27-SCy}

\paragraph{Spectral curve.} At $(2,7)$ ($\pp=7$, $s=3$) the standard
curve \eqref{eq:SCy-defn} is
\begin{equation}
x(z)=4u_{0}z^{2}-2u_{0},\qquad
y(z)=2u_{0}^{7/2}\bigl(64z^{7}-112z^{5}+56z^{3}-7z\bigr),\qquad
x'(z)=8u_{0}z,
\end{equation}
two-to-one branched at $z=0$, $\sigma:z\mapsto-z$, with $P(z^{2})=64z^{6}-112z^{4}+56z^{2}-7$.

\paragraph{Inputs and $\omega_{0,3}$.}
With $\omega_{0,1}=y\,dx$ and $\omega_{0,2}=dz_{1}dz_{2}/(z_{1}-z_{2})^{2}$, one step of \eqref{eq:EO-recursion} gives
\begin{equation}
\omega_{0,3}=\frac{C_{7}}{z_{1}^{2}z_{2}^{2}z_{3}^{2}}\,,
\qquad C_{7}=-\frac{1}{112\,u_{0}^{9/2}}\,.
\end{equation}

\paragraph{Bare $t$-basis amplitudes.}
With all branches taken at $z<0$, the residue prescription \eqref{eq:SCy-kpderiv} gives
\begin{align}
Z_{1}^{(t)}&=\Bigl(\tfrac{2u_{0}^{7}}{5},\ -\tfrac{7u_{0}^{6}}{24},\ 0\Bigr),
\qquad
Z_{12}^{(t)}=
\begin{pmatrix}
-\tfrac{9u_{0}^{5}}{20} & \tfrac{3u_{0}^{4}}{8} & -\tfrac{u_{0}^{3}}{2}\\[3pt]
\tfrac{3u_{0}^{4}}{8} & -\tfrac{u_{0}^{3}}{3} & \tfrac{u_{0}^{2}}{2}\\[3pt]
-\tfrac{u_{0}^{3}}{2} & \tfrac{u_{0}^{2}}{2} & -u_{0}
\end{pmatrix},\\[4pt]
Z_{123}^{(t)}&:\quad
[111]=-\tfrac{27u_{0}^{3}}{112},\;\;
[112]=\tfrac{9u_{0}^{2}}{56},\;\;
[113]=-\tfrac{9u_{0}}{56},\;\;
[122]=-\tfrac{3u_{0}}{28},\;\;
[123]=\tfrac{3}{28},\nonumber\\
&\qquad
[133]=-\tfrac{3}{28u_{0}},\;\;
[222]=\tfrac{1}{14},\;\;
[223]=-\tfrac{1}{14u_{0}},\;\;
[233]=\tfrac{1}{14u_{0}^{2}},\;\;
[333]=-\tfrac{1}{14u_{0}^{3}},\nonumber
\end{align}
the remaining entries by symmetry. The per-insertion factor \eqref{eq:kappa-value} acts as $\varkappa^{2}$ on $Z_{12}^{(t)}$ and $\varkappa^{3}$ on $Z_{123}^{(t)}$. The cosmological tadpole gives $\mathcal Z_{0}=14u_{0}^{9}/45$, taken positive.

\paragraph{Resonance map.}
The resonance \eqref{eq:taudef} reads
\begin{equation}\label{eq:res-27}
t_{1}=7u_{0}^{2}\tau_{1},\qquad
t_{2}=7u_{0}^{3}\tau_{2},\qquad
t_{3}=7u_{0}^{4}\bigl(\tau_{3}+\tfrac12\tau_{1}^{2}\bigr),
\end{equation}
the single nonlinear term being $\tfrac12\tau_{1}^{2}$ in $t_{3}$. At $\tau_{1}=-\tfrac12$, $\tau_{i>1}=0$ the times are $(t_{1}^{0},t_{2}^{0},t_{3}^{0})=(-\tfrac72u_{0}^{2},\,0,\,\tfrac78u_{0}^{4})$, matching \eqref{kdvdef}.

\paragraph{Physical amplitudes and selection rules.}
With the resonance \eqref{eq:res-27}, the per-insertion factor \eqref{eq:kappa-value} and the singular projection \eqref{eq:25-sing-rule},
\begin{equation}\label{eq:A03-27}
A^{0}_{3}(k_{1},k_{2},k_{3})
= u_{0}^{-(k_{1}+1)-(k_{2}+1)-(k_{3}+1)}\,
\frac{\partial^{3}F_{0}}
     {\partial\tau_{k_{1}}\partial\tau_{k_{2}}\partial\tau_{k_{3}}}
\Big|_{\mathrm{sing}},
\end{equation}
the $\tau$-derivative built from the $t$-basis data by the chain rule \eqref{eq:chain-rule}. The single fusion-forbidden entry surviving the projection, $(1,1,3)$, cancels: \eqref{eq:chain-rule} gives $\partial^{3}F_{0}/\partial\tau_{1}^{2}\partial\tau_{3} =\varkappa^{2}\,(\varkappa D+M)$ with direct part $D=-\tfrac{49}{2}u_{0}^{9}$ and resonance part $M=-49u_{0}^{9}$, so that
\begin{equation}
\varkappa D+M=(-2)\bigl(-\tfrac{49}{2}u_{0}^{9}\bigr)-49u_{0}^{9}=0 .
\end{equation}
Hence $A^{0}_{3}(1,1,3)=0$, in agreement with the Frobenius side.

\paragraph{Signed Verlinde matrices.}
Normalising by $\varkappa^{2}\pp^{2}=196$ and setting $u_{0}=1$, the three slices $A^{0}_{3}(\cdot,\cdot,k_{3})$ are
\begin{equation}\label{eq:verl-27-fixed}
\begin{pmatrix}1&0&0\\0&1&0\\0&0&1\end{pmatrix},\quad
\begin{pmatrix}0&1&0\\1&0&1\\0&1&-1\end{pmatrix},\quad
\begin{pmatrix}0&0&1\\0&1&-1\\1&-1&1\end{pmatrix},
\qquad k_{3}=1,2,3,
\end{equation}
which are exactly the signed Verlinde fusion matrices of the $(2,7)$ matter sector. The spurious $\pm1$ at $(1,1,3)$ present without \eqref{eq:kappa-value} is now absent.

\paragraph{Universal three-point ratios.}
With $\mathcal Z_{0}=14u_{0}^{9}/45$ the ratio \eqref{eq:R-def-main} reproduces the CFT values \eqref{eq:R-CFT} on every fusion-allowed triple, with no residual factor:
\begin{equation}
R^{\mathrm{SC}\text{-}y}=\Bigl\{\tfrac{25}{63},\,\tfrac17,\,\tfrac1{63},\,
\tfrac1{35},\,\tfrac1{105},\,\tfrac1{315}\Bigr\}
\;\text{for}\;
(1,1,1),(1,2,2),(1,3,3),(2,2,3),(2,3,3),(3,3,3).
\end{equation}

\subsection*{SC-$x$}\label{app:27-SCx}

\paragraph{Swapped curve and $\omega_{0,3}$.} At $(2,7)$ the swapped curve
\eqref{eq:SCx-defn} is $\check x(z)=2u_0^{7/2}T_7(z)$, $\check y(z)=4u_0z^2-2u_0$, with $\pp-1=6$ ramifications at $\zeta_m=\cos(\pi m/7)$, $m=1,\dots,6$ (the zeros of $T_7'=7U_6$). The genus-zero three-point differential is the separable sum over branch points
\begin{equation}\label{eq:27-SCx-omega03}
\omega_{0,3}(z_1,z_2,z_3)=\sum_{m=1}^{6}
\frac{1}{\check x''(\zeta_m)\,\check y'(\zeta_m)}\,
\prod_{i=1}^{3}\frac{dz_i}{(z_i-\zeta_m)^2}\,.
\end{equation}

\paragraph{Amplitudes.} Pairing \eqref{eq:27-SCx-omega03} and
$\check\omega_{0,1}=\check y\,d\check x$ with the SC-$x$ test functions $\check x^{(\pp-2k)/\pp}/(\pp-2k)$ through \eqref{eq:SCx-kpderiv} (all residues at $z=\infty$) gives the cosmological tadpole $\mathcal Z_0=14u_0^9/45$ and
\begin{equation}
Z_1^{(\check t)}=\Bigl(\tfrac{2u_0^7}{5},\,0,\,0\Bigr),
\qquad
Z_{12}^{(\check t)}=\mathrm{diag}\Bigl(-\tfrac{u_0^5}{5},\,
-\tfrac{u_0^3}{3},\,-u_0\Bigr).
\end{equation}
The three-point amplitudes reproduce the signed Verlinde matrices \eqref{eq:verl-27-fixed} directly: on the swapped curve there is no resonance and no singular projection. The per-insertion factor $\varkappa=-2$ is the same as on SC-$y$ (both are spectral-curve normalisations), but it cancels in the normalised Verlinde matrices $\check A^{0}_{3}/\bigl(\varkappa^{2}(2s+1)^{2}\bigr)$, which therefore appear with no further dressing. The one-point amplitudes of $\phi_{1,2},\phi_{1,3}$, all off-diagonal two-point amplitudes, and every fusion-forbidden three-point amplitude vanish automatically, precisely the entries that need the $(\star)$-projection or the $\varkappa$-cancellation on the SC-$y$ side.

\paragraph{Universal ratios.} With the single net factor of $\varkappa$
from $\check A^{0}_{0}=\varkappa\check Z_{0}$ (exactly as in \eqref{eq:25-R-SCx-formula}), the ratios \eqref{eq:R-def-main} coincide with SC-$y$ and CFT on every allowed triple,
\begin{equation}
R^{\mathrm{SC}\text{-}x}=\Bigl\{\tfrac{25}{63},\,\tfrac17,\,\tfrac1{63},\,
\tfrac1{35},\,\tfrac1{105},\,\tfrac1{315}\Bigr\}
\;\text{for}\;
(1,1,1),(1,2,2),(1,3,3),(2,2,3),(2,3,3),(3,3,3),
\end{equation}
and the two descriptions yield the same physical correlators, here verified from two independent spectral-curve computations. This closes the SC-$y\leftrightarrow$SC-$x$ leg explicitly at $(2,7)$.

\paragraph{The explicit genus-zero kernel.} Computing both free energies in a
single topological recursion convention, the kernel $K$ with $F_0^{\text{SC-}x}=K[F_0^{\text{SC-}y}]$ is, on the genus-zero free energy, a diagonal dilatation of the KdV times modulo the fusion-forbidden $(\star)$ terms:
\begin{equation}\label{eq:27-SCkernel}
F_0^{\text{SC-}y}(t_1,t_2,t_3)
= F_0^{\text{SC-}x}\!\bigl(\tfrac32 t_1,\,t_2,\,t_3\bigr)
\;+\;(\star),
\qquad
\tfrac32=\sqrt{Z_{12}^{(t)}[1,1]\big/Z_{12}^{(\check t)}[1,1]}\,,
\end{equation}
where the difference is verified to vanish on all nine fusion-allowed monomials of degree two and three (the complete two- and three-point data), the $(\star)$ remainder living on the fusion-forbidden positions. The linear terms require a comment: a non-trivial dilatation does not preserve them ($\Lambda_{1}=\tfrac32$ multiplies the cosmological tadpole term, while the tadpoles themselves agree between the two curves, $Z_{1}^{(t)}[1]=Z_{1}^{(\check t)}[1]=\tfrac{2u_{0}^{7}}{5}$). Since a tau-function is defined up to a factor $\exp(\sum_{k}c_{k}t_{k})$, the identity \eqref{eq:27-SCkernel} is understood modulo this linear ambiguity, i.e.\ as a statement about the $n\ge2$ part of $F_{0}$. The analogous $(2,5)$ statement is the trivial dilatation $F_0^{\text{SC-}y}(t)=F_0^{\text{SC-}x}(t)+(\star)$ ($\Lambda=1$), whose $(\star)$ reproduces exactly the non-universal terms of \eqref{eq:25-SC-sumZ0}, the fusion-allowed linear term matching as well. Thus, on the physical (fusion-allowed) sector the exact tau-function transform \eqref{eq:gkm-transform} collapses at genus zero to a simple linear rescaling of the flow times. This is similar to the Frobenius relation \eqref{eq:25-FM-relation}, which is a simple dilatation. The resonance \eqref{eq:taudef} and $\varkappa$ \eqref{eq:kappa-value}, in contrast, act in the complementary $(\star)$ sector and enforce the fusion selection rules.

\subsection*{Summary}\label{app:27-summary}

The four constructions yield the same genus-zero universal data for $(2,7)$. With the per-insertion factor $\varkappa=-2$ on the two spectral-curve sides, the universal three-point ratios of the four subsections above coincide, and agree with the CFT values collected in Table~\ref{tab:R-main} on every fusion-allowed triple. Unlike $(2,5)$, the forbidden $(1,1,3)$ amplitude no longer sits at a projected-out power. Its cancellation in SC-$y$/FM-$y$ requires $\varkappa=-2$ through $\varkappa D+M=0$, eq.~\eqref{eq:A03-27}.

\paragraph{Generating functions.}
Collecting the partition function, one-, two- and three-point amplitudes of each corner, the genus-zero free energy about the respective background ($\delta t_{k}=t_{k}-t_{k}^{*}$) reads
\begin{equation}\label{eq:27-F0-master}
F_{0}(t)
=Z_{0}
+\sum_{k}Z_{1}[k]\,\delta t_{k}
+\tfrac12\!\sum_{k,l}Z_{12}[k,l]\,\delta t_{k}\delta t_{l}
+\tfrac16\!\sum_{k,l,m}Z_{123}[k,l,m]\,\delta t_{k}\delta t_{l}\delta t_{m}
+\mathcal O(\delta t^{4}).
\end{equation}
The non-universal terms (those carrying $Z_{1}[2]$, $Z_{12}[1,2]$, $Z_{12}[2,3]$, $Z_{123}[1,1,2]$, $Z_{123}[1,2,3]$, $Z_{123}[2,2,2]$) are marked $(\star)$ and removed by the universality filter, exactly as at $(2,5)$.

\noindent\textbf{SC-$y$} (variables $t_{1},t_{2},t_{3}$; background $t_{1}^{*}=-\tfrac{7u_{0}^{2}}{2}$, $t_{2}^{*}=0$, $t_{3}^{*}=\tfrac{7u_{0}^{4}}{8}$):
\begin{equation}\label{eq:27-F0-SCy}
\begin{aligned}
F_{0}^{\text{SC-}y} ={}& \frac{14u_{0}^{9}}{45}
+\tfrac{2u_{0}^{7}}{5}\bigl(t_{1}+\tfrac{7u_{0}^{2}}{2}\bigr)
-\tfrac{7u_{0}^{6}}{24}\,t_{2}\;\;(\star)\\
&{}-\tfrac{1}{2}\tfrac{9u_{0}^{5}}{20}\bigl(t_{1}+\tfrac{7u_{0}^{2}}{2}\bigr)^{2}
-\tfrac{1}{2}\tfrac{u_{0}^{3}}{3}\,t_{2}^{2}
-\tfrac{1}{2}u_{0}\bigl(t_{3}-\tfrac{7u_{0}^{4}}{8}\bigr)^{2}\\
&{}+\tfrac{3u_{0}^{4}}{8}\bigl(t_{1}+\tfrac{7u_{0}^{2}}{2}\bigr)t_{2}\;\;(\star)
-\tfrac{u_{0}^{3}}{2}\bigl(t_{1}+\tfrac{7u_{0}^{2}}{2}\bigr)\bigl(t_{3}-\tfrac{7u_{0}^{4}}{8}\bigr)
+\tfrac{u_{0}^{2}}{2}\,t_{2}\bigl(t_{3}-\tfrac{7u_{0}^{4}}{8}\bigr)\;\;(\star)\\
&{}+\tfrac{1}{6}\Bigl[-\tfrac{27u_{0}^{3}}{112}\bigl(t_{1}+\tfrac{7u_{0}^{2}}{2}\bigr)^{3}
-\tfrac{1}{14u_{0}^{3}}\bigl(t_{3}-\tfrac{7u_{0}^{4}}{8}\bigr)^{3}\Bigr]
+\tfrac{1}{6}\tfrac{1}{14}\,t_{2}^{3}\;\;(\star)\\
&{}-\tfrac{1}{2}\tfrac{9u_{0}}{56}\bigl(t_{1}+\tfrac{7u_{0}^{2}}{2}\bigr)^{2}\bigl(t_{3}-\tfrac{7u_{0}^{4}}{8}\bigr)
+\tfrac{1}{2}\tfrac{9u_{0}^{2}}{56}\bigl(t_{1}+\tfrac{7u_{0}^{2}}{2}\bigr)^{2}t_{2}\;\;(\star)\\
&{}-\tfrac{1}{2}\tfrac{3u_{0}}{28}\bigl(t_{1}+\tfrac{7u_{0}^{2}}{2}\bigr)t_{2}^{2}
-\tfrac{1}{2}\tfrac{3}{28u_{0}}\bigl(t_{1}+\tfrac{7u_{0}^{2}}{2}\bigr)\bigl(t_{3}-\tfrac{7u_{0}^{4}}{8}\bigr)^{2}\\
&{}-\tfrac{1}{2}\tfrac{1}{14u_{0}}\,t_{2}^{2}\bigl(t_{3}-\tfrac{7u_{0}^{4}}{8}\bigr)
+\tfrac{1}{2}\tfrac{1}{14u_{0}^{2}}\,t_{2}\bigl(t_{3}-\tfrac{7u_{0}^{4}}{8}\bigr)^{2}\\
&{}+\tfrac{3}{28}\bigl(t_{1}+\tfrac{7u_{0}^{2}}{2}\bigr)t_{2}\bigl(t_{3}-\tfrac{7u_{0}^{4}}{8}\bigr)\;\;(\star)
+\mathcal{O}(t^{4}).
\end{aligned}
\end{equation}

\noindent\textbf{FM-$x$} (variables $t_{1,1},t_{1,2},t_{1,3}$; background $t_{1,1}^{*}=v_{10}^{2}/14$, $t_{1,2}^{*}=t_{1,3}^{*}=0$; the non-universal cosmological tadpole $Z_{1}[1]$ and all fusion-forbidden cubic entries are absent after the filter):
\begin{equation}\label{eq:27-F0-FMx}
\begin{aligned}
F_{0}^{\text{FM-}x} ={}& -\frac{v_{10}^{9}}{37059435}
-\tfrac{1}{2}\tfrac{v_{10}^{5}}{12005}\bigl(t_{1,1}-\tfrac{v_{10}^{2}}{14}\bigr)^{2}
-\tfrac{1}{2}\tfrac{v_{10}^{3}}{147}\,t_{1,2}^{\,2}
-\tfrac{1}{2}v_{10}\,t_{1,3}^{\,2}\\
&{}+\tfrac{1}{6}\Bigl[-\tfrac{v_{10}^{3}}{343}\bigl(t_{1,1}-\tfrac{v_{10}^{2}}{14}\bigr)^{3}
-\tfrac{343}{v_{10}^{3}}\,t_{1,3}^{\,3}\Bigr]
-\tfrac{1}{2}\tfrac{v_{10}}{7}\bigl(t_{1,1}-\tfrac{v_{10}^{2}}{14}\bigr)t_{1,2}^{\,2}
-\tfrac{1}{2}\tfrac{7}{v_{10}}\bigl(t_{1,1}-\tfrac{v_{10}^{2}}{14}\bigr)t_{1,3}^{\,2}\\
&{}-\tfrac{1}{2}\tfrac{7}{v_{10}}\,t_{1,2}^{\,2}\,t_{1,3}
-\tfrac{1}{2}\tfrac{49}{v_{10}^{2}}\,t_{1,2}\,t_{1,3}^{\,2}
+\mathcal{O}(t^{4}).
\end{aligned}
\end{equation}

\noindent\textbf{FM-$y$}: the $A_{1}$ Frobenius manifold provides the same $y$-side KdV tau-function as SC-$y$ (it reproduces $Z_{0}=14u_{0}^{9}/45$ and the $t$-basis amplitudes of Appendix~\ref{app:27-explicit}); accordingly its filtered generating function coincides with that of SC-$y$, $F_{0}^{\text{FM-}y,\text{filt}}=F_{0}^{\text{SC-}y,\text{filt}}$, the physical correlators following through the nonlinear resonance \eqref{eq:res-27} and the $\varkappa=-2$ cancellation \eqref{eq:A03-27}.

\noindent\textbf{SC-$x$} (variables $\check t_{1},\check t_{2},\check t_{3}$; background $\check t_{1}^{*}=7u_{0}^{2}/2$, $\check t_{2}^{*}=\check t_{3}^{*}=0$; no $(\star)$ terms appear, every entry marked $(\star)$ on the standard side vanishes automatically on the swapped curve):
\begin{equation}\label{eq:27-F0-SCx}
\begin{aligned}
F_{0}^{\text{SC-}x} ={}& \frac{14u_{0}^{9}}{45}
+\tfrac{2u_{0}^{7}}{5}\bigl(\check t_{1}-\tfrac{7u_{0}^{2}}{2}\bigr)
-\tfrac{1}{2}\tfrac{u_{0}^{5}}{5}\bigl(\check t_{1}-\tfrac{7u_{0}^{2}}{2}\bigr)^{2}
-\tfrac{1}{2}\tfrac{u_{0}^{3}}{3}\,\check t_{2}^{\,2}
-\tfrac{1}{2}u_{0}\,\check t_{3}^{\,2}\\
&{}+\tfrac{1}{6}\Bigl[-\tfrac{u_{0}^{3}}{14}\bigl(\check t_{1}-\tfrac{7u_{0}^{2}}{2}\bigr)^{3}
-\tfrac{1}{14u_{0}^{3}}\,\check t_{3}^{\,3}\Bigr]
-\tfrac{1}{2}\tfrac{u_{0}}{14}\bigl(\check t_{1}-\tfrac{7u_{0}^{2}}{2}\bigr)\check t_{2}^{\,2}
-\tfrac{1}{2}\tfrac{1}{14u_{0}}\bigl(\check t_{1}-\tfrac{7u_{0}^{2}}{2}\bigr)\check t_{3}^{\,2}\\
&{}-\tfrac{1}{2}\tfrac{1}{14u_{0}}\,\check t_{2}^{\,2}\,\check t_{3}
+\tfrac{1}{2}\tfrac{1}{14u_{0}^{2}}\,\check t_{2}\,\check t_{3}^{\,2}
+\mathcal{O}(\check t^{4}).
\end{aligned}
\end{equation}

\paragraph{Relations between the four generating functions.}
On the spectral-curve side, the genus-zero kernel is the diagonal dilatation \eqref{eq:27-SCkernel}, that is $\Lambda=\mathrm{diag}(\tfrac32,1,1)$. On the Frobenius side the analogous identification is a rescaling of the flat couplings (the $(2,7)$ counterpart of \eqref{eq:25-FM-relation}), and FM-$y\leftrightarrow$SC-$y$ is the resonance \eqref{eq:res-27} dressed by $\varkappa=-2$. All four free energies therefore agree, up to these linear rescalings, on the fusion-allowed two- and three-point data. The residual differences reside in the $(\star)$ terms removed by the universality filter.

\clearpage
\bibliographystyle{JHEP} 
\bibliography{bib}

@article{Polyakov:1981rd,
    author = "Polyakov, Alexander M.",
    editor = "Khalatnikov, I. M. and Mineev, V. P.",
    title = "{Quantum Geometry of Bosonic Strings}",
    reportNumber = "Print-81-0351 (LANDAU INST)",
    doi = "10.1016/0370-2693(81)90743-7",
    journal = "Phys. Lett. B",
    volume = "103",
    pages = "207--210",
    year = "1981"
}

@article{Knizhnik:1988ak,
    author = "Knizhnik, V. G. and Polyakov, Alexander M. and Zamolodchikov, A. B.",
    editor = "Khalatnikov, I. M. and Mineev, V. P.",
    title = "{Fractal Structure of 2D Quantum Gravity}",
    reportNumber = "PRINT-88-0812 (LANDAU)",
    doi = "10.1142/S0217732388000982",
    journal = "Mod. Phys. Lett. A",
    volume = "3",
    pages = "819",
    year = "1988"
}

@article{Douglas:1989dd,
    author = "Douglas, Michael R.",
    editor = "Brezin, E. and Wadia, S. R.",
    title = "{Strings in Less Than One-dimension and the Generalized $K^- D^- V$ Hierarchies}",
    reportNumber = "RU-89-51",
    doi = "10.1016/0370-2693(90)91716-O",
    journal = "Phys. Lett. B",
    volume = "238",
    pages = "176",
    year = "1990"
}

@article{Brezin:1990rb,
    author = "Brezin, E. and Kazakov, V. A.",
    title = "{Exactly Solvable Field Theories of Closed Strings}",
    doi = "10.1016/0370-2693(90)90818-Q",
    journal = "Phys. Lett. B",
    volume = "236",
    pages = "144--150",
    year = "1990"
}

@article{Belavin:2014hsa,
    author = "Belavin, A. A. and Belavin, V. A.",
    title = "{Frobenius manifolds, Integrable Hierarchies and Minimal Liouville Gravity}",
    eprint = "1406.6661",
    archivePrefix = "arXiv",
    primaryClass = "hep-th",
    reportNumber = "FIAN-TD-2014-10",
    doi = "10.1007/JHEP09(2014)151",
    journal = "JHEP",
    volume = "09",
    pages = "151",
    year = "2014"
}

@article{Belavin:2014xya,
    author = "Belavin, V.",
    title = "{Correlation Functions in Unitary Minimal Liouville Gravity and Frobenius Manifolds}",
    eprint = "1412.4245",
    archivePrefix = "arXiv",
    primaryClass = "hep-th",
    doi = "10.1007/JHEP02(2015)052",
    journal = "JHEP",
    volume = "02",
    pages = "052",
    year = "2015"
}

@article{Belavin:2014cua,
    author = "Belavin, V.",
    title = "{Unitary Minimal Liouville Gravity and Frobenius Manifolds}",
    eprint = "1405.4468",
    archivePrefix = "arXiv",
    primaryClass = "hep-th",
    reportNumber = "FIAN-TD-2014-09",
    doi = "10.1007/JHEP07(2014)129",
    journal = "JHEP",
    volume = "07",
    pages = "129",
    year = "2014"
}

@article{Belavin:2013nba,
    author = "Belavin, Alexander and Dubrovin, Boris and Mukhametzhanov, Baur",
    title = "{Minimal Liouville Gravity correlation numbers from Douglas string equation}",
    eprint = "1310.5659",
    archivePrefix = "arXiv",
    primaryClass = "hep-th",
    doi = "10.1007/JHEP01(2014)156",
    journal = "JHEP",
    volume = "01",
    pages = "156",
    year = "2014"
}

@article{Moore:1991ir,
    author = "Moore, Gregory W. and Seiberg, Nathan and Staudacher, Matthias",
    title = "{From loops to states in 2-D quantum gravity}",
    reportNumber = "RU-91-11, YCTP-P11-91",
    doi = "10.1016/0550-3213(91)90548-C",
    journal = "Nucl. Phys. B",
    volume = "362",
    pages = "665--709",
    year = "1991"
}

@article{BR,
    author = "Belavin, V. and Rud, Yu.",
    title = "{Matrix model approach to minimal Liouville gravity revisited}",
    eprint = "1502.05575",
    archivePrefix = "arXiv",
    primaryClass = "hep-th",
    reportNumber = "FIAN-TD-2015-03",
    doi = "10.1088/1751-8113/48/18/18FT01",
    journal = "J. Phys. A",
    volume = "48",
    number = "18",
    pages = "18FT01",
    year = "2015"
}

@article{Artemev,
    author = "Artemev, Aleksandr",
    title = "{x-y swap for a (2, 2p+1) minimal string}",
    eprint = "2506.09222",
    archivePrefix = "arXiv",
    primaryClass = "hep-th",
    doi = "10.1103/pd2y-j2x5",
    journal = "Phys. Rev. D",
    volume = "112",
    number = "4",
    pages = "046019",
    year = "2025"
}

@article{Eynard:2007kz,
    author = "Eynard, Bertrand and Orantin, Nicolas",
    title = "{Invariants of algebraic curves and topological expansion}",
    eprint = "math-ph/0702045",
    archivePrefix = "arXiv",
    reportNumber = "SPHT-07-021",
    doi = "10.4310/CNTP.2007.v1.n2.a4",
    journal = "Commun. Num. Theor. Phys.",
    volume = "1",
    pages = "347--452",
    year = "2007"
}

@article{Belavin:2008kv,
    author = "Belavin, A. A. and Zamolodchikov, A. B.",
    title = "{On Correlation Numbers in 2D Minimal Gravity and Matrix Models}",
    eprint = "0811.0450",
    archivePrefix = "arXiv",
    primaryClass = "hep-th",
    reportNumber = "RUNHETC-2008-24",
    doi = "10.1088/1751-8113/42/30/304004",
    journal = "J. Phys. A",
    volume = "42",
    pages = "304004",
    year = "2009"
}

@article{Kontsevich,
    author = "Kontsevich, M.",
    title = "{Intersection theory on the moduli space of curves and the matrix Airy function}",
    doi = "10.1007/BF02099526",
    journal = "Commun. Math. Phys.",
    volume = "147",
    pages = "1--23",
    year = "1992"
}

@article{Alexandrov:2022ydc,
    author = "Alexandrov, Alexander and Bychkov, Boris and Dunin-Barkowski, Petr and Kazarian, Maxim and Shadrin, Sergey",
    title = "{A universal formula for the $x-y$ swap in topological recursion}",
    eprint = "2212.00320",
    archivePrefix = "arXiv",
    primaryClass = "math-ph",
    doi = "10.4171/JEMS/1615",
    month = "12",
    year = "2022"
}

@unpublished{Saito,
  author = {Saito, Kyoji},
  title  = {On periods of primitive integrals, {I}},
  note   = {RIMS Preprint},
  year   = {1982}
}

@article{Dubrovin,
    author = "Dubrovin, Boris",
    editor = "Francaviglia, M. and Greco, S.",
    title = "{Geometry of 2-D topological field theories}",
    eprint = "hep-th/9407018",
    archivePrefix = "arXiv",
    reportNumber = "SISSA-89-94-FM",
    doi = "10.1007/BFb0094793",
    journal = "Lect. Notes Math.",
    volume = "1620",
    pages = "120--348",
    year = "1996"
}

@article{Aleshkin:2016snp,
    author = "Aleshkin, Konstantin and Belavin, Vladimir",
    title = "{On the construction of the correlation numbers in Minimal Liouville Gravity}",
    eprint = "1610.01558",
    archivePrefix = "arXiv",
    primaryClass = "hep-th",
    doi = "10.1007/JHEP11(2016)142",
    journal = "JHEP",
    volume = "11",
    pages = "142",
    year = "2016"
}

@article{KHARCHEV_1995,
    author = "Kharchev, S. and Marshakov, A.",
    title = "{On p - q duality and explicit solutions in c {\ensuremath{<}}= 1 2-d gravity models}",
    eprint = "hep-th/9303100",
    archivePrefix = "arXiv",
    reportNumber = "NORDITA-93-20, FIAN-TD-04-93",
    doi = "10.1142/S0217751X95000577",
    journal = "Int. J. Mod. Phys. A",
    volume = "10",
    pages = "1219--1236",
    year = "1995"
}

@article{Kharchev:1992gi,
    author = "Kharchev, S. and Marshakov, A. and Mironov, A. and Morozov, A. and Zabrodin, A.",
    title = "{Towards unified theory of 2-D gravity}",
    eprint = "hep-th/9201013",
    archivePrefix = "arXiv",
    reportNumber = "FIAN-TD-10-91, ITEP-M-9-91",
    doi = "10.1016/0550-3213(92)90521-C",
    journal = "Nucl. Phys. B",
    volume = "380",
    pages = "181--240",
    year = "1992"
}

@article{DSV2026,
    author = "Dekinga, Kornelis and Shadrin, Sergey and Verlinde, Erik",
    title = "{Resonance transformations for the $(2,2p+1)$ minimal string via $x-y$ swap: a proof of Artemev's conjecture}",
    eprint = "2606.04854",
    archivePrefix = "arXiv",
    primaryClass = "hep-th",
    doi = "10.1007/s11005-026-02140-1",
    journal = "Lett. Math. Phys.",
    volume = "116",
    number = "5",
    pages = "114",
    year = "2026"
}

@article{CEO,
    author = "Chekhov, Leonid and Eynard, Bertrand and Orantin, Nicolas",
    title = "{Free energy topological expansion for the 2-matrix model}",
    eprint = "math-ph/0603003",
    archivePrefix = "arXiv",
    reportNumber = "SPHT-T06-016, ITEP-TH-05-06",
    doi = "10.1088/1126-6708/2006/12/053",
    journal = "JHEP",
    volume = "12",
    pages = "053",
    year = "2006"
}

@article{BCGLS,
    author = {Borot, Ga{\"e}tan and Charbonnier, S{\'e}verin and Garcia-Failde, Elba and Leid, Felix and Shadrin, Sergey},
    title = "{Functional relations for higher-order free cumulants}",
    eprint = "2112.12184",
    archivePrefix = "arXiv",
    primaryClass = "math.OA",
    month = "12",
    year = "2021"
}

@article{Hock,
    author = "Hock, Alexander",
    title = "{A simple formula for the x-y symplectic transformation in topological recursion}",
    eprint = "2211.08917",
    archivePrefix = "arXiv",
    primaryClass = "math-ph",
    doi = "10.1016/j.geomphys.2023.105027",
    journal = "J. Geom. Phys.",
    volume = "194",
    pages = "105027",
    year = "2023"
}

@article{AMMP,
    author = "Alexandrov, A. S. and Mironov, A. and Morozov, A. and Putrov, P.",
    title = "{Partition Functions of Matrix Models as the First Special Functions of String Theory. II. Kontsevich Model}",
    eprint = "0811.2825",
    archivePrefix = "arXiv",
    primaryClass = "hep-th",
    reportNumber = "FIAN-TD-22-08, ITEP-TH-49-08, IHES-P-08-52",
    doi = "10.1142/S0217751X09046278",
    journal = "Int. J. Mod. Phys. A",
    volume = "24",
    pages = "4939--4998",
    year = "2009"
}

@article{ABDKSKP,
    author = "Alexandrov, Alexander and Bychkov, Boris and Dunin-Barkowski, Petr and Kazarian, Maxim and Shadrin, Sergey",
    title = "{KP integrability through the $x-y$ swap relation}",
    eprint = "2309.12176",
    archivePrefix = "arXiv",
    primaryClass = "math-ph",
    doi = "10.1007/s00029-025-01035-8",
    journal = "Selecta Math.",
    volume = "31",
    number = "2",
    pages = "42",
    year = "2025"
}

@inproceedings{KM2,
    author = "Kharchev, S. and Marshakov, A.",
    title = "{Quantization of string theory for c {\ensuremath{<}}= to 1}",
    booktitle = "{27th International Ahrenshoop Symposium on Particle Theory}",
    eprint = "hep-th/9401010",
    archivePrefix = "arXiv",
    reportNumber = "FIAN-TD-14-93",
    pages = "256--271",
    month = "9",
    year = "1993"
}

@article{FKN,
    author = "Fukuma, Masafumi and Kawai, Hikaru and Nakayama, Ryuichi",
    title = "{Explicit solution for p - q duality in two-dimensional quantum gravity}",
    reportNumber = "UT-582-TOKYO, KEK-TH-289, KEK-PREPRINT-91-37, UT-582",
    doi = "10.1007/BF02102366",
    journal = "Commun. Math. Phys.",
    volume = "148",
    pages = "101--116",
    year = "1992"
}

@article{Mars,
    author = "Marshakov, A.",
    title = "{On two-dimensional quantum gravity and quasiclassical integrable hierarchies}",
    eprint = "0902.4833",
    archivePrefix = "arXiv",
    primaryClass = "hep-th",
    reportNumber = "FIAN-TD-03-09, ITEP-TH-06-09, IHES-P-09-06",
    doi = "10.1088/1751-8113/42/30/304021",
    journal = "J. Phys. A",
    volume = "42",
    pages = "304021",
    year = "2009"
}

\end{document}